\documentclass[runningheads,envcountsame]{llncs}

\usepackage{amsmath,amssymb}
\usepackage{bbm}
\usepackage[numbers,sort&compress]{natbib}
\usepackage{enumitem}

\usepackage{nicefrac}

\usepackage[colorlinks,linkcolor=blue,citecolor=blue]{hyperref}

\makeatletter
\newcommand{\customlabel}[2]{%
   \protected@write \@auxout {}{\string \newlabel {#1}{{#2}{\thepage}{#2}{#1}{}} }%
   \hypertarget{#1}{#2}
}
\makeatother

\spnewtheorem{observation}[theorem]{Observation}{\bfseries}{\itshape}

\usepackage[capitalize,nameinlink]{cleveref}
\crefname{algocf}{Algorithm}{Algorithms}
\crefname{observation}{Observation}{Observations}

\usepackage{thm-restate}

\usepackage[algoruled,vlined]{algorithm2e}
\usepackage{setspace}

\usepackage{tikz}
\tikzstyle{normalNode}=[draw=black, fill=black, circle, minimum size=2mm, inner sep=0mm]
\tikzstyle{normalEdge}=[black, thick, >=stealth]

\newcommand{\st}{\;:\;}
\newcommand{\prob}[1]{\mathbf{Pr}\left[#1\right]}
\newcommand{\sets}{\mathcal{P}} 
\newcommand{\starsuff}{\eqref{eq:cond}-sufficient}

\newcommand{\indicator}[1]{\mathbbm{1}_{#1}}

\newcommand{\incidence}[2][]{%
\ifthenelse { \equal {#1} {} }%
{\mathbbm{1}_{#2}}%
{\mathbbm{1}_{#1 \in #2}}%
}

\newcommand{\polystar}{Y^{\star}}
\newcommand{\polyplus}{Y^{+}}
\newcommand{\polyProj}{Y_{\pi}}
\newcommand{\polyDist}{Z_{\pi}}

\begin{document}

\title{Decomposing Probability Marginals\\ Beyond Affine Requirements\thanks{%
Proofs of results marked with $(\clubsuit)$ can be found in the appendix.%
}}

\titlerunning{Decomposing Probability Marginals Beyond Affine Requirements}

\author{Jannik Matuschke}

\institute{KU Leuven}

\maketitle

\begin{abstract}
    Consider the triplet $(E, \sets, \pi)$, where $E$ is a finite ground set,~$\sets \subseteq 2^E$ is a collection of subsets of $E$ and \mbox{$\pi : \sets \rightarrow [0,1]$} is a \emph{requirement function}. 
    Given a vector of \emph{marginals} $\rho \in [0, 1]^E$, our goal is to find a distribution for a random subset $S \subseteq E$ such that $\prob{e \in S} = \rho_e$ for all~\mbox{$e \in E$} and $\prob{P \cap S \neq \emptyset} \geq \pi_P$ for all $P \in \sets$,
    or to determine that no such distribution exists.
    
    Generalizing results of Dahan, Amin, and Jaillet~\citep{dahan2021probability}, we devise a generic decomposition algorithm that solves the above problem when provided with a suitable sequence of \emph{admissible support candidates (ASCs)}.
    We show how to construct such ASCs for numerous settings, including \emph{supermodular requirements}, Hoffman-Schwartz-type \emph{lattice polyhedra}~\citep{hoffman1978lattice}, and \emph{abstract networks} where $\pi$ fulfils a conservation law.
    The resulting algorithm can be carried out efficiently when~$\sets$ and~$\pi$ can be accessed via appropriate oracles. For any system allowing the construction of ASCs, our results imply a simple polyhedral description of the set of marginal vectors for which the decomposition problem is feasible. 
    Finally, we characterize \emph{balanced hypergraphs} as the systems~$(E, \sets)$ that allow the \emph{perfect decomposition} of any marginal vector $\rho \in [0,1]^E$, i.e., where we can always find a distribution reaching the highest attainable probability~$\prob{P \cap S \neq \emptyset} = \min  \left\{ \sum_{e \in P} \rho_e, 1\right\}$ for all~$P \in \sets$. 
\end{abstract}

\section{Introduction}
\label{sec:introduction}

Given a set system $(E, \sets)$ on a finite ground set $E$ with $\sets \subseteq 2^E$ and a \emph{requirement function} $\pi : \sets \rightarrow (-\infty,1]$, consider the polytope
$$\textstyle \polyDist := \left\{z \in [0, 1]^{2^E} \st \sum_{S \subseteq E} z_S = 1 \text{ and } \sum_{S : S \cap P \neq \emptyset} z_S \geq \pi_P \ \forall\; P \in \sets\right\},$$
which corresponds to the set of all probability distributions over $2^E$ such that the corresponding random subset $S \subseteq E$ hits each $P \in \sets$ with probability at least its requirement value $\pi_P$.\footnote{Note that we can assume $\pi_P \in [0, 1]$ without loss of generality in the definition of~$\polyDist$, but we allow negative values for notational convenience in later parts of the paper.}
We are interested in describing the projection of~$\polyDist$ to the corresponding marginal probabilities on $E$, i.e.,
$$\textstyle \polyProj := \left\{\rho \in [0, 1]^E \st \exists z \in \polyDist \text{ with } \rho_e = \sum_{S \subseteq E : e \in S} z_S \ \forall\, e \in E \right\}.$$
For $\rho \in Y_\pi$, we call any $z \in \polyDist$ with $\rho_e = \sum_{S \subseteq E : e \in S} z_S$ for all $e \in E$ a \emph{feasible decomposition of~$\rho$ for $(E, \sets, \pi)$}.
Note that every $\rho \in \polyProj$ fulfils
\begin{align}
    \textstyle \sum_{e \in P} \rho_e \geq \pi_P  \qquad \forall\, P \in \mathcal{P} \tag{$\star$}\label{eq:cond}
\end{align}
because $\sum_{S : S \cap P \neq \emptyset} z_S \leq \sum_{e \in P} \rho_e$ for any feasible decomposition $z$ of $\rho$. Hence
\begin{align*}
    \polyProj \;\subseteq\; \polystar := \left\{\rho \in [0, 1]^E \st \rho \text{ fulfils } \eqref{eq:cond}\right\}.
\end{align*}
We say that $(E, \sets, \pi)$ is \emph{\starsuff} if~$\polyProj = \polystar$.
Our goal is to identify classes of such {\starsuff} systems, along with corresponding decomposition algorithms that, given $\rho \in \polystar$, find a feasible decomposition of $\rho$. 
Using such decomposition algorithms, we can reduce optimization problems over $\polyDist$ whose objectives and other constraints can be expressed via the marginals to optimization problems over~$\polystar$, yielding an exponential reduction in dimension.

\subsection{Motivation}

Optimization problems over~$\polyDist$ and polytopes with a similar structure arise, e.g., in the context of \emph{security games}.
In such a game, a defender selects a random subset $S \subseteq E$ of resources to inspect while an attacker selects a strategy $P \in \sets$, balancing their utility from the attack against the risk of detection~(which occurs if $P \cap S \neq \emptyset$).
Indeed, the decomposition setting described above originates from the work of \mbox{\citet*{dahan2021probability}}, who used it to describe the set of mixed Nash equilibria for such a security game using a compact LP formulation when the underlying system is {\starsuff}.

Two further application areas of marginal decomposition are \emph{randomization in robust or online optimization}, which is often used to overcome pessimistic worst-case scenarios~\citep{kawase2019submodular,kobayashi2017randomized,kawase2019randomized,matuschke2018robust}, 
and \emph{social choice and mechanism design}, where randomization is frequently used to satisfy otherwise irreconcilable axiomatic requirements~\cite{brandl2016consistent} and where decomposition results in various flavors are applied, e.g., to define auctions via interim allocations~\citep{border1991implementation,gopalan2018public}, to improve load-balancing in school choice~\citep{demeulemeester2023pessimist}, and to turn approximation algorithms into truthful mechanisms~\citep{kraft2014fast,lavi2011truthful}.
In \cref{sec:applications}, we discuss several applications from these three areas, including different security games, a robust randomized coverage problem, and committee election with diversity constraints. We show how the structures for which we establish \eqref{eq:cond}-sufficiency in this paper arise naturally in these applications and imply efficient algorithms for these settings.

\subsection{Previous Results}
 
As mentioned above, \citet{dahan2021probability} introduced the decomposition problem described above to characterize mixed Nash equilibria of a network security game played on $(E, \sets)$.
They observed that such equilibria can be described by a compact LP formulation if $(E, \sets, \pi)$ is {\starsuff} for all requirements $\pi$ of the \emph{affine} form 
\begin{align}
    \textstyle \pi_P = 1 - \sum_{e \in P} \mu_e \quad \forall\, P \in \sets \tag{A}\label{eq:affine}
\end{align}
for some $\mu \in [0, 1]^E$.
They showed that this is indeed the case when $E$ is the set of edges of a directed acyclic graph (DAG) and $\sets$ the set of $s$-$t$-paths in this DAG and provide an polynomial-time (in $|E|$) algorithm for computing feasible decompositions in this case.
\citet{matuschke2023decomposition-full} extended this result by providing an efficient decomposition algorithm for \emph{abstract networks}, a generalization of the system of $s$-$t$-paths in a (not necessarily acyclic) digraph; see \cref{sec:abstract-networks} for a definition.
He also showed that a system $(E, \sets, \pi)$ is {\starsuff} for all affine requirement functions $\pi$ if and only if the system has the \emph{weak max-flow/min-cut property}, i.e., the polyhedron $\{y \in \mathbb{R}_+^E \st \sum_{e \in P} y_e \geq 1 \ \forall\, P \in \sets\}$ is integral.

While the affine setting \eqref{eq:affine} is well-understood, little is known for the case of more general requirement functions. A notable exception is the \emph{conservation law} studied by~\citet{dahan2021probability}, again for the case of directed acyclic graphs:
\begin{align}
    {\pi_P + \pi_Q = \pi_{P \times_e Q} + \pi_{Q \times_e P}} \quad \forall\, P, Q \in \mathcal{P}, e \in P \cap Q, \tag{C}\label{eq:conservation}
\end{align}
where $P \times_e Q$ for two paths $P, Q \in \sets$ containing a common edge $e \in P \cap Q$ denotes the path consisting of the prefix of $P$ up to $e$ and the suffix of $Q$ starting with $e$.
\citet{dahan2021probability} established \eqref{eq:cond}-sufficiency for requirements fulfilling~\eqref{eq:conservation} in DAGs by providing another combinatorial decomposition algorithm.
It was later observed in~\citep{matuschke2023decomposition-full} and, independently in a different context in~\citep{cela2023linear}, that \eqref{eq:conservation} for DAGs is in fact equivalent to~\eqref{eq:affine}.
However, this equivalence no longer holds for the natural generalization of \eqref{eq:conservation} to arbitrary digraphs.

\subsection{Contribution and Structure of this Paper}
\label{sec:contribution}

In this article, we present an algorithmic framework for computing feasible decompositions of marginal vectors fulfilling~\eqref{eq:cond} for a wide range of set systems and requirement functions, going beyond the affine setting~\eqref{eq:affine}.
Our algorithm, described in \cref{sec:algorithm}, iteratively adds a so-called \emph{admissible support candidate~(ASC)} to the constructed decomposition.
The definition of ASCs is based on a transitive dominance relation on $\sets$, which has the property that a decomposition of $\rho \in \polystar$ is feasible for $(E, \sets, \pi)$ if and only if it is feasible for the restriction of the system to non-dominated sets.

Our algorithmic framework can be seen as a generalization of Dahan et al.'s~\citep{dahan2021probability}~Algorithm~1 for requirements fulfilling \eqref{eq:conservation} in DAGs.  
An important novelty which allows us to establish \eqref{eq:cond}-sufficiency for significantly more general settings is the use of the dominance relation and the definition of ASCs, which are more flexible than the properties implicitly used in \citep{dahan2021probability}. 
A detailed comparison of the two algorithms can be found in \cref{app:dahan-et-al}.

To establish correctness of our algorithm for a certain class of systems, which also implies \eqref{eq:cond}-sufficiency for those systems, it suffices to show the existence of an ASC in each iteration of the algorithm.
We assume that the set $E$ is of small cardinality and given explicitly, while $\sets$ might be large (possibly exponential in $|E|$) and is accessed by an appropriate oracle.
To obtain an polynomial run-time of our algorithm in $|E|$, it suffices to show that the following two tasks can be carried out in polynomial time in $|E|$:
\begin{enumerate}[label=(\roman*),align=left]
    \item In each iteration, construct an ASC.
    \item Given~\mbox{$\rho \in [0, 1]^E$}, either assert $\rho \in \polystar$ or find a maximum violated inequality of \eqref{eq:cond}, i.e., $P \in \sets$ maximizing $\pi_P - \sum_{e \in P} \rho_e > 0$.
\end{enumerate}
We prove the existence and computability of admissible sets for a variety of settings, which we describe in the following.

\paragraph{Supermodular Requirements.}
A basic example for which our algorithm implies~\eqref{eq:cond}-sufficiency is the case where $\sets = 2^E$ and $\pi$ is a \emph{supermodular} function, i.e.,~$\pi_{P \cap Q} + \pi_{P \cup Q}  \geq \pi_{P} + \pi_{Q}$ for all $P, Q \in \sets$.
In \cref{sec:contrapolymatroid}, we show the existence of ASCs for this setting and observe that both (i) and (ii) can be solved when $\pi$ is given by a \emph{value oracle} that given $P \in \sets$ returns $\pi_P$.

\paragraph{Abstract Networks under Weak Conservation of Requirements.} 
We prove \eqref{eq:cond}-sufficiency for the case that $(E, \sets)$ is an abstract network and $\pi$ fulfils a relaxed version of the conservation law~\eqref{eq:conservation} introduced by \citet{hoffman1974generalization}.
Such systems generalize systems of $s$-$t$-paths in digraphs, capturing some of their essential properties that suffice to obtain results such as Ford and Fulkerson's~\citep{ford1956maximal} max-flow/min-cut theorem or Dijkstra's~\citep{dijkstra1959note} shortest-path algorithm; see \cref{sec:abstract-networks} for a formal definition and an in-depth discussion.
In particular, our results generalize the results of~\citet{dahan2021probability} for DAGs under~\eqref{eq:conservation} to arbitrary digraphs.

\paragraph{Lattice Polyhedra.}
We also study the case where $\sets \subseteq 2^E$ is a \emph{lattice}, i.e., a partially ordered set in which each pair of incomparable elements have a unique maximum common lower bound, called \emph{meet} and a unique minimum common upper bound, called \emph{join}, and where $\pi$ is supermodular with respect to these meet and join operations.
\citet{hoffman1978lattice} showed that under two additional assumptions on the lattice, called \emph{submodularity} and \emph{consecutivity}, the system defined by \eqref{eq:cond} and $\rho \geq 0$ is totally dual integral (the corresponding polyhedron, which is the dominant of $\polystar$, is called \emph{lattice polyhedron}).
These polyhedra generalize (contra-)polymatroids and describe, e.g., $r$-cuts in a digraph~\citep{frank1999increasing} or paths in $s$-$t$-planar graphs~\citep{matuschke2010lattices}.
When $\pi$ is monotone with respect to the partial order on~$\sets$, a two-phase (primal-dual) greedy algorithm introduced by \citet{kornblum1978greedy} and later generalized by \citet{frank1999increasing} can be used to efficiently optimize linear functions over lattice polyhedra using an oracle that returns maxima of sublattices. 
We show the existence and computability of admissible sets under the same assumptions by carefully exploiting the structure of extreme points implicit in the analysis of the Kornblum-Frank algorithm; 
see \cref{sec:lattice-polyhedra} for complete formal definitions and an in-depth discussion of these results.

\paragraph{Perfect Decompositions and Balanced Hypergraphs.}
We call a set system $(E, \sets)$ \emph{decomposition-friendly} if it is {\starsuff} for all requirement funtions $\pi$. 
Note that $(E, \sets)$ is decomposition-friendly if and only if every $\rho \in [0, 1]^E$ has a feasible decomposition for $(E, \sets, \pi^{\rho})$, where $\pi^{\rho}_P := \min \left\{\sum_{e \in P} \rho_e, 1\right\}$ for $P \in \sets$. We call such a decomposition \emph{perfect}, as it simultaneously reaches the maximum intersection probability attainable under $\rho$ for each $P \in \sets$.
In \cref{sec:balanced} we show that $(E, \sets)$ is decomposition-friendly if and only if it is a \emph{balanced hypergraph}, a set system characterized by the absence of certain odd-length induced cycles.

\subsection{Notation and Preliminaries}

For $m \in \mathbb{N}$, we use the notation $[m]$ to denote the set $\{1, \dots, m\}$. 
Moreover, we use the notation $\indicator{A}$ to indicate whether expression $A$ is true ($\indicator{A} = 1$) or false ($\indicator{A} = 0$).
We will further make use of the following observation. 

\begin{lemma}[{\cite[Lemma 3]{matuschke2023decomposition-full}}]\label{lem:feasible-relax}
    There is an algorithm that given $\rho \in [0, 1]^E$ and $z \in \polyDist$ with $\sum_{S : e \in S} z_S \leq \rho_e$ for all $e \in E$, computes a feasible decomposition of $\rho$ in time polynomial in $|E|$ and $|\{S \subseteq E \st z_S > 0\}|$.
\end{lemma}

\section{Decomposition Algorithm}
\label{sec:algorithm}

We describe a generic algorithm that is able to compute feasible decompositions of marginals for a wide range of systems.
The algorithm makes use of a dominance relation defined in \cref{sec:dominance}.
We describe the algorithm in \cref{sec:algorithm-description} and state the conditions under which it is guaranteed to produce a feasible decomposition.
In \cref{sec:algorithm-examples}, we provide a simple yet relevant example where these conditions are met.
Finally, we prove correctness of the algorithm in \cref{sec:algorithm-analysis}.

\subsection{The Relation $\sqsubseteq_{\pi,\rho}$ and Admissible Support Candidates}
\label{sec:dominance}

For $P, Q \in \sets$ we write $P \sqsubseteq_{\pi, \rho} Q$ if $\pi_P \leq \pi_Q - \sum_{e \in Q \setminus P} \rho_e$ and $\pi_P < \pi_Q$, or if $P = Q$.
We say that $P$ is \emph{non-dominated} with respect to $\pi$ and $\rho$ in $\sets' \subseteq \sets$ if $P \in \sets'$ and there exists no $Q \in \sets' \setminus \{P\}$ with $P \sqsubseteq_{\pi,\rho} Q$.

\begin{restatable}[$\clubsuit$]{lemma}{restateLemDominanceTransitive}\label{lem:dominance-transitive}
    The relation $\sqsubseteq_{\pi,\rho}$ is a partial order.
    In particular, for any $\sets' \subseteq \sets$, there exists at least one $P'$ that is non-dominated in $\sets'$. 
\end{restatable}

As we will see in the analysis below, it suffices to ensure $\sum_{S : S \cap P} z_S \geq \pi_P$ for non-dominated $P \in \sets$ to construct a feasible decomposition.
This motivates the following definition.
A set $S \subseteq E$ is an \emph{admissible support candidate (ASC) for $\pi$ and $\rho$} if the following three conditions are fulfilled:
\begin{enumerate}[label=(S\arabic*),align=left]
    \item $S \subseteq E_{\rho} := \{e \in E \st \rho_e > 0\}$. 
    \label{prop:cut:positive}
    \item $|S \cap P| \leq 1$ for all $P \in \sets_{\pi,\rho}^= := \left\{Q \in \sets \st \sum_{e \in Q} \rho_e = \pi_Q\right\}$. 
    \label{prop:cut:S-tight-sets}
    \item $|P \cap S| \geq 1$ for all non-dominated (w.r.t.~$\pi$ and $\rho$) $P$ in $\{Q \in \sets \st \pi_Q > 0\}$. 
    
    \label{prop:cut:S-hits-necessary-sets}
\end{enumerate}
We now present an algorithm, that when provided with a sequence of ASCs computes a feasible decomposition for $\rho \in \polystar$.

\subsection{The Algorithm}
\label{sec:algorithm-description}

The algorithm constructs a decomposition by iteratively selecting an ASC~$S$ for a requirement function $\bar{\pi}$ and a marginal vector $\bar{\rho}$, which can be thought of as residuals of the original requirements and marginals, respectively, with $\bar{\pi} = \pi$ and $\bar{\rho} = \rho$ initially.
It shifts a probability mass of 
$$\textstyle \varepsilon_{\bar{\pi},\bar{\rho}}(S) := \min\, \left\{  
\min_{e \in S} \bar{\rho}_e,\ 
\max_{P \in \sets} \bar{\pi}_P,\ 
\delta_{\bar{\pi},\bar{\rho}}(S)  
\right\}$$
to $S$, where $\delta_{\bar{\pi},\bar{\rho}}(S) := \inf_{P \in \sets : |P \cap S| > 1}
    \frac{ \bar{\pi}_P - \sum_{e \in P} \bar{\rho}_e }{ 1 - |P \cap S| }$.
Intuitively, $\varepsilon_{\bar{\pi},\bar{\rho}}(S)$ corresponds to the maximum amount of probability mass that can be shifted to the set $S$ without losing feasibility of the remaining marginals for the remaining requirements.
The residual marginals $\bar{\rho}$ are reduced by $\varepsilon_{\bar{\pi},\bar{\rho}}(S)$ for all $e \in S$, and so are the requirements of all $P \in \sets$ (including those $P$ with $P \cap S = \emptyset$).

\medskip

\begin{algorithm}[H]
  \caption{Generic Decomposition Algorithm}\label{alg:decomposition}
  \setstretch{1.1}
  Initialize $\bar{\pi} := \pi$, $\bar{\rho} := \rho$.\\
  Initialize $z_{\emptyset} = 1$ and $z_S := 0$ for all $S \subseteq E$ with $S \neq \emptyset$.\\
  \While{$\max_{P \in \sets} \bar{\pi}_P > 0$\vspace{0.1cm}}{
    Let $S$ be an ASC for $\bar{\pi}$ and $\bar{\rho}$.\\  
    Let $\varepsilon := \varepsilon_{\bar{\pi},\bar{\rho}}(S)$.\\
    Set $z_S := z_S + \varepsilon$ and $z_{\emptyset} := z_{\emptyset} - \varepsilon$.\\
    Set $\bar{\rho}_e := \bar{\rho}_e - \varepsilon$ for all $e \in S$.\\
    Set $\bar{\pi}_P := \bar{\pi}_P - \varepsilon$ for all $P \in \sets$.\\
  }
  Apply \cref{lem:feasible-relax} to $z$ to obtain a feasible decomposition $z'$ of $\rho$.\\
  \Return $z'$
\end{algorithm}

\medskip

Our main result establishes that the algorithm returns a feasible decomposition after a polynomial number of iterations, if an ASC for $\bar{\pi}$ and $\bar{\rho}$ exists in every iteration.
To show that a certain system is {\starsuff}, it thus suffices to establish the existence of the required ASCs.

\begin{theorem}
    \label{thm:decomposition-algo}
    Let $(E, \sets)$ be a set system and $\pi : \sets \rightarrow (-\infty, 1]$. Let $\rho \in \polystar$.
    If there exists an ASC for $\bar{\pi}$ and $\bar{\rho}$ in every iteration of \cref{alg:decomposition}, then the algorithm terminates after $\mathcal{O}(|E|^2)$ iterations and returns a feasible decomposition of $\rho$ for $(E, \sets, \pi)$.
\end{theorem}

Note that \cref{thm:decomposition-algo} implies that \cref{alg:decomposition} can be implemented to run in time $\mathcal{O}(\mathcal{T}|E|^2)$, when provided with an oracle that computes the required ASCs along with the corresponding values of $\varepsilon_{\bar{\pi},\bar{\rho}}(S)$ in time $\mathcal{T}$.%
\footnote{In particular, note that $\varepsilon_{\bar{\pi},\bar{\rho}}(S)$ can be computed using 
at most $|S|$ iterations of the discrete Newton algorithm if we can solve problem~(ii) from \cref{sec:contribution}, i.e., the maximum violated inequality problem for $\polystar$.}
Before we prove \cref{thm:decomposition-algo}, we first provide an example to illustrate its application.

\subsection{Basic Example: Supermodular Requirements}
\label{sec:algorithm-examples}
\label{sec:contrapolymatroid}

Consider the case that $\sets = 2^E$ and $\pi$ is supermodular, i.e.,  for all \mbox{$P, Q \in \sets$} it holds that $\pi_{P \cap Q} + \pi_{P \cup Q} \geq \pi_P + \pi_Q$.
Note that if $\pi$ is supermodular, then~$\bar{\pi}$ is supermodular throughout \cref{alg:decomposition}, as subtracting a constant does not affect supermodularity.
Moreover, we show in \cref{sec:algorithm-analysis} that $\bar{\rho}$ fulfils \eqref{eq:cond} for~$\bar{\pi}$ throughout the algorithm.
To apply \cref{alg:decomposition}, it thus suffices to show existence of an ASC when~$\rho \in \polystar$ and $\pi$ is supermodular.
To obtain the ASC, we define $Q := \bigcup_{P \in \sets^=_{\pi,\rho}} P$ and distinguish two cases:
If~$Q \cap E_{\rho} = \emptyset$, we let $S' := E_{\rho} \setminus Q$. Otherwise, we let $S' :=  (E_{\rho} \setminus Q) \cup \{e_Q\}$ for an arbitrary $e_Q \in Q \cap E_{\rho}$.

\begin{lemma}
    If $\sets = 2^E$, $\pi$ is supermodular, and $\rho \in \polystar$, then $S'$ is an ASC.
\end{lemma}
\begin{proof}
    Note that $S$ fulfils~\ref{prop:cut:positive} and \ref{prop:cut:S-tight-sets} by construction because $P \subseteq Q$ for all~$P \in \sets^=_{\pi,\rho}$.
    Moreover, for any $P \in \sets$, either $P \cap S' \neq \emptyset$, or $P \cap E_{\rho} \subseteq Q \cap E_{\rho} = \emptyset$ and hence $\pi_P \leq 0$, or~\mbox{$P \cap E_{\rho} \subseteq Q \setminus \{e_Q\}$}.
    In the third case, we make use of the fact that, by standard uncrossing arguments, $Q \in \sets^=_{\pi,\rho}$ .
    From this we obtain $\pi_P \leq \sum_{e \in P} \rho_e = \sum_{e \in Q \cap P} \rho_e = \pi_Q - \sum_{e \in Q \setminus P} \rho_e$ and therefore~$P \sqsubseteq_{\pi,\rho} Q$ (note that $e_Q \in Q \setminus P$ and hence $\pi_P < \pi_Q$). 
    Thus $S'$ fulfils \ref{prop:cut:S-hits-necessary-sets}. \qed
\end{proof}

We remark that both the described ASC and maximum violated inequalities of $\polystar$ can be found in polynomial time using submodular function minimization~\citep{schrijver2003combinatorial} when $\pi$ is given by a value oracle, that given $P$ returns $\pi_P$.

\subsection{Analysis (Proof of \cref{thm:decomposition-algo})}
\label{sec:algorithm-analysis}

Throughout this section we assume that $(E, \sets, \pi)$ and $\rho$ fulfil the conditions of the \cref{thm:decomposition-algo}. In particular, $\rho \in \polystar$ and in each iteration of the algorithm there exists an ASC.
We show that under these conditions the while-loop terminates after $\mathcal{O}(|E|^2)$ iterations (\cref{lem:C-iterations}) and that after termination of the loop, $z \in \polyDist$ (\cref{lem:C-covering}) and $\sum_{S : e \in S} z_S \leq \rho_e$ for all $e \in E$ (\cref{lem:C-invariants}\ref{inv:rho} for $k = \ell$).
This implies that \cref{lem:feasible-relax} can indeed be applied to $z$ in the algorithm to obtain a feasible decomposition of $\rho$, thus proving \cref{thm:decomposition-algo}.

We introduce the following notation. Let $S^{(i)}$ and $\varepsilon^{(i)}$ denote the set $S$ and the value of $\varepsilon$ chosen in the $i$th iteration of the while loop in the algorithm. 
Let further $\pi^{(i)}$ and $\rho^{(i)}$ denote the values of $\bar{\pi}$ and $\bar{\rho}$ at the beginning of the $i$th iteration (in particular, $\pi^{(1)} = \pi$ and $\rho^{(1)} = \rho$).
Let~$K \subseteq \mathbb{N}$ denote the set of iterations of the while loop.
If the algorithm terminates, $K = \{1, \dots, \ell\}$, where $\ell \in \mathbb{N}$ denotes the number of iterations.
In that case, let $\rho^{(\ell+1)}$ and $\pi^{(\ell+1)}$ denote the state of $\bar{\rho}$ and $\bar{\pi}$ after termination.

Using this notation, we can establish the following three invariants, which follow directly from the construction of $\rho^{(i)}$ and $\varepsilon^{(i)}$ in the algorithm and the defining properties of the ASC $S^{(i)}$.

\begin{restatable}[$\clubsuit$]{lemma}{restateLemCInvariants}\label{lem:C-invariants}
    For all $k \in K$, the following statements hold true:\\[-15pt]
    \begin{enumerate}[label=\textup{(}\alph*\textup{)},align=left]
        \item $\rho^{(k+1)}_e = \rho_e - \sum_{i = 1}^{k} \incidence[e]{S^{(i)}} \cdot \varepsilon^{(i)} \geq 0$ for all $e \in E$,
        \label{inv:rho}
        \item $\sum_{e \in P} \rho^{(k + 1)}_e \geq \pi^{(k + 1)}_P = \pi_P - \sum_{i = 1}^{k} \varepsilon^{(i)}$ for all $P \in \sets$, and
        \label{inv:pi}
        \item $S^{(k)} \neq \emptyset$ and $\varepsilon^{(k)} > 0$.
        \label{inv:S-eps}
    \end{enumerate}
\end{restatable}

The next lemma shows that the while loop indeed terminates after $\mathcal{O}(|E|^2)$ iterations.
Its proof follows from the fact that in every non-final iteration $k \in K$, there is an element $e \in S^{(k)}$ for which the value of $\bar{\rho}_e$ drops to $0$, or there are two elements $e, e' \in S^{(k)}$ such that $e, e' \in P$ for some $P \in \sets^=_{\pi^{(k+1)},\rho^{(k+1)}}$.
It can be shown that the same pair $e, e'$ cannot appear in two distinct iterations of the latter type, from which we obtain the following bound.

\begin{restatable}[$\clubsuit$]{lemma}{restateLemCIterations}\label{lem:C-iterations}
    The while loop in~\cref{alg:decomposition} terminates after at most $\binom{|E|}{2} + |E|$ iterations, i.e., $K = \{1, \dots, \ell\}$ with $\ell \leq \binom{|E|}{2} + |E|$.
\end{restatable}

The termination criterion of the while loop implies the following lemma.

\begin{restatable}[$\clubsuit$]{lemma}{restateLemCTotalEpsilon}\label{lem:C-total-epsilon}
    It holds that $\sum_{i=1}^{\ell} \varepsilon^{(i)} =  \max_{P \in \sets} \pi_P$.
\end{restatable}

Finally, we can use the properties of the ASCs~$S^{(k)}$ to show that $z \in \polyDist$.

\begin{restatable}{lemma}{restateLemCCovering}\label{lem:C-covering}
    After termination of the while loop, it holds that $z \in \polyDist$.
\end{restatable}

\begin{proof}
    Note that $z_S = \sum_{i=1}^{\ell} \indicator{S = S^{(i)}} \cdot \varepsilon^{(i)} \geq 0$ for $S \subseteq E$ with $S \neq \emptyset$ and that $z_{\emptyset} = 1 - \sum_{i=1}^{\ell} \varepsilon^{(i)} \geq 0$, where the nonnegativity follows from \cref{lem:C-invariants}\ref{inv:S-eps} and \cref{lem:C-total-epsilon} with $\max_{P \in \sets} \pi_P \leq 1$, respectively.
    This also implies $\sum_{S \subseteq E} z_S = 1$.
    
    We will prove that  $\sum_{i = k}^{\ell} \mathbbm{1}_{P \cap S^{(i)} \neq \emptyset} \cdot \varepsilon^{(i)} \geq \pi^{(k)}_P$ for all $k \in [\ell+1]$ and $P \in \sets$, which, for $k = 1$, implies
    $\sum_{S : P \cap S \neq \emptyset} z_S  = \sum_{i = 1}^{\ell} \indicator{P \cap S^{(i)} \neq \emptyset} \cdot \varepsilon^{(i)} \geq \pi^{(1)}_P = \pi_P$ and hence
    \mbox{$z \in \polyDist$}.
    We prove the above statement by induction on $k$, starting from $k = \ell + 1$ and going down to $k = 1$. For the base case $k = \ell + 1$, observe that the left-hand side is $0$ and $\pi^{\ell+1}_P \leq 0$ by termination criterion of the while loop.

For the induction step, let $k \in [\ell]$, assuming that the statement is already established for $k + 1$ and let $P \in \sets$. 
We distinguish two cases.
\begin{itemize}
    \item Case $P \cap S^{(k)} \neq \emptyset$: We can apply the induction hypothesis to obtain\\
    \mbox{$\sum_{i = k}^{\ell} \mathbbm{1}_{P \cap S^{(i)} \neq \emptyset} \cdot \varepsilon^{(i)} 
    = \varepsilon^{(k)} + \sum_{i = k + 1}^{\ell} \mathbbm{1}_{P \cap S^{(i)} \neq \emptyset} \cdot \varepsilon^{(i)}
    \geq \varepsilon^{(k)} + \pi^{(k+1)}_P = \pi^{(k)}_P$}.
    \item Case $P \cap S^{(k)} = \emptyset$: 
    If $\pi^{(k)}_P \leq 0$ then the desired statement follows from $\varepsilon^{(i)} > 0$ for all $i \in [\ell]$ by \cref{lem:C-invariants}\ref{inv:S-eps}. 
    Thus, we can assume $\pi^{(k)}_P > 0$.
    By property~\ref{prop:cut:S-hits-necessary-sets}, there is $Q \in \sets$ with $P \sqsubseteq_{\pi^{(k)}, \rho^{(k)}} Q$ and $Q \cap S^{(k)} \neq \emptyset$.
    Hence we can apply the induction step proven in the first case to $Q$, yielding
    $\sum_{i = k}^{\ell} \mathbbm{1}_{Q \cap S^{(i)} \neq \emptyset} \cdot \varepsilon^{(i)} \geq \pi^{(k)}_Q$.
    From this, we conclude that
    \begin{align*}
        \textstyle \sum_{i = k}^{\ell} \mathbbm{1}_{P \cap S^{(i)} \neq \emptyset} \cdot \varepsilon^{(i)} &
        \textstyle 
        \;\geq\; \sum_{i = k}^{\ell} \mathbbm{1}_{P \cap Q \cap S^{(i)} \neq \emptyset} \cdot \varepsilon^{(i)}\\
        & \textstyle  \;\geq\; \pi^{(k)}_Q - \sum_{i = k}^{\ell} \mathbbm{1}_{(Q \setminus P) \cap S^{(i)} \neq \emptyset} \cdot \varepsilon^{(i)}\\
        & \textstyle  \;\geq\; \pi^{(k)}_Q - \sum_{e \in Q \setminus P} \rho^{(k)}_e \;\geq\; \pi^{(k)}_P,
    \end{align*}
    where the first and second inequality use $\varepsilon^{(i)} > 0$ by \cref{lem:C-invariants}\ref{inv:S-eps}, the third inequality uses
    $\rho^{(k)}_e = \sum_{i = k}^{\ell} \incidence[e]{ S^{(i)}} \cdot \varepsilon^{(i)}$ by \cref{lem:C-invariants}\ref{inv:rho}
    and the final inequality uses $P \sqsubseteq_{\pi^{(k)}, \rho^{(k)}} Q$.\hfill$\square$
\end{itemize}
\end{proof}

\section{Abstract Networks Under Weak Conservation Law}
\label{sec:abstract-networks}

A set system $(E, \sets)$ is an \emph{abstract network} if for every $P \in \sets$, there is an order $\preceq_P$ of the elements in $P$ and for every $P, Q \in \sets$ and $e \in P \cap Q$ there is $R \in \sets$ with $R \subseteq \{p \in P \st p \preceq_P e\} \cup \{q \in Q \st e \preceq q\}$.
We use the notation $P \times_e Q$ to denote an arbitrary but fixed choice of such an $R$.
Note that the definition of abstract networks does not impose any requirements on the order $\preceq_{P \times_e Q}$. In particular, it does not need to be consistent with $\preceq_P$ and $\preceq_Q$.

Abstract networks were introduced by \citet{hoffman1974generalization} in an effort to encapsulate the essential properties of systems of paths in classic networks that enable the proof of \citeauthor{ford1956maximal}'s~\citep{ford1956maximal} max-flow/min-cut theorem.
Indeed, the set of $s$-$t$-paths in a digraph constitutes a special case of an abstract network (however, see~\citep{kappmeier2014abstract} for examples of abstract networks that do not arise in this way) and the elements of $\sets$ are therefore also referred to as \emph{abstract paths}. 
The \emph{maximum weighted abstract flow (MWAF)} problem and the \emph{minimum weighted abstract cut (MWAC)} problem correspond to the linear programs
\begin{align*}
    \begin{array}{rrlr}
        \max \  & \sum_{P \in \sets} \pi_P  & x_P & \\
        \text{s.t.} \  & \sum_{P : e \in P} x_P & \leq u_e & \forall e \in E \\
        & x & \geq 0 &
    \end{array}
    \quad
    \begin{array}{rrlr}
        \min \  & \sum_{e \in E} u_e  & y_e & \\
        \text{s.t.} \  & \sum_{e \in P} y_e & \geq \pi_P & \forall P \in \sets \\
        & y & \geq 0 &
    \end{array}
\end{align*}
where $u \in \mathbb{R}_+^E$ is a capacity vector and $\pi$ determines the reward for each unit of flow send along the abstract path $P \in \sets$.

\citet{hoffman1974generalization} proved that MWAC is totally dual integral, if the reward function $\pi$ fulfils the following weak conservation law:
\begin{align}
    {\pi_P + \pi_Q \geq \pi_{P \times_e Q} + \pi_{Q \times_e P}} \quad \forall\, P, Q \in \mathcal{P}, e \in P \cap Q. \tag{C'}\label{eq:conservation-weak}
\end{align}
\citet{mccormick1996polynomial} complemented this result by a combinatorial algorithm for solving MWAF when $\pi \equiv 1$. This was extended by \citet{martens2008polynomial} to a combinatorial algorithm for solving MWAF with arbitrary $\pi$ fulfilling \eqref{eq:conservation-weak} when $\pi$ is given a separation oracle for the constraints of MWAC.

A combinatorial algorithm for marginal decomposition in abstract networks under affine requirements \eqref{eq:affine} based on a generalization of Dijkstra's shortest-path algorithm is presented in \citep{matuschke2023decomposition-full}.
Here, we show \eqref{eq:cond}-sufficiency for the more general setting~\eqref{eq:conservation-weak} by showing that ASCs can be constructed in this setting.

\begin{theorem}\label{thm:abstract-networks}
    Let $(E, \sets)$ be an abstract network, let $\pi$ fulfil \eqref{eq:conservation-weak} and $\rho \in \polystar$.
    Then $S := \{e \in E_{\rho} \st \text{there is no } P \in \sets_{\pi,\rho}^= \text{ and } p \in P \cap E_{\rho} \text{ with } p \prec_P e\}$ is an ASC for $\pi$ and $\rho$.
\end{theorem}

\begin{proof}[\cref{thm:abstract-networks}]
We use the notation $(P, e) := \{p \in P \st p \prec_P e\}$ and  $[e, P] := \{p \in P \st e \preceq_P e\}$ for $P \in \sets$ and $e \in P$.
Note that $S$ fulfils \ref{prop:cut:positive} and \ref{prop:cut:S-tight-sets} by construction.
It remains to show that $S$ also fulfils \ref{prop:cut:S-hits-necessary-sets}.
For this let $Q \in \sets$ be non-dominated with $\pi_Q > 0$ and assume by contradiction $Q \cap S = \emptyset$.

Note that $Q \cap E_{{\rho}} \neq \emptyset$ because $\sum_{e \in Q} {\rho}_e \geq {\pi}_Q > 0$. Let $q := \min_{\preceq_Q} Q \cap E_{{\rho}}$.
Observe that $q \notin S$ by our assumption, and hence, by construction of $S$, there must be $r \in E_{{\rho}}$ and $R \in \sets^=_{{\pi},{\rho}}$ such that $r \prec_{R} q$.

Let $Q' := R \times_{q} Q$ and $R' := Q \times_{q} R$.
Note that $R' \cap E_{\rho} \subseteq [q, R]$ because $(Q, q) \cap E_{\rho} = \emptyset$ by choice of $q$ as $\prec_q$-minimal element in $Q \cap E_{\rho}$.
Using \eqref{eq:cond}, we obtain $\sum_{e \in [q, R]} \rho_e \geq \sum_{e \in R'} \rho_e \geq \pi_{R'}$.
We conclude that
\begin{align*}
    \textstyle \pi_{Q'} + \sum_{e \in [q, R]} \rho_e \;\geq\; \pi_{Q'} + \pi_{R'} \;\geq\; \pi_{Q} + \pi_{R} \;=\; \pi_{Q} + \sum_{e \in R} \rho_e,
\end{align*}
where the second inequality follows from \eqref{eq:conservation-weak} and the final identity is due to the fact that $\pi_R = \sum_{e \in R} \rho_e$ because $R \in \sets_{\pi,\rho}$.
Subtracting $\sum_{e \in [q, R]} \rho_e$ on both sides yields
    $\pi_{Q'} \geq \pi_{Q} + \sum_{e \in (R, q)} \rho_e$.
    
Using $Q' \setminus Q \subseteq (R, q)$ by construction of $Q'$, we obtain $\pi_{Q'} \geq \pi_Q + \sum_{e \in R \setminus Q} \rho_e$.
Note further that $\pi_{Q'} > \pi_{Q}$ because $r \in (R, q) \cap E_{\rho}$ and hence $\sum_{e \in (R, q)} \rho_e > 0$.
We conclude that $Q \sqsubseteq_{\pi,\rho} Q'$, a contradiction $Q$ being non-dominated.
\qed    
\end{proof}

We remark that the corresponding ASCs and hence feasible decompositions can be commputed in polynomial time in $|E|$ if the abstract network is given via an oracle that solve the maximum violated inequality problem for $\polystar$ and returns $\pi_P$ and $\prec_P$ for the corresponding $P \in \sets$.

\section{Lattice Polyhedra}
\label{sec:lattice-polyhedra}

We now consider the case where $\sets$ is equipped with a partial order $\preceq$ so that $(\sets, \preceq)$ is a \emph{lattice}, i.e., the following two properties are fulfilled for all $P, Q \in \sets$:
\begin{itemize}
    \item The set $\{R \in \sets \st R \preceq P,\, R \preceq Q\}$ has a unique maximum w.r.t.~$\preceq$, denoted by $P \wedge Q$ and called the \emph{meet of $P$ and $Q$}.
    \item The set $\{R \in \sets \st R \succeq P,\, R \succeq Q\}$ has a unique minimum w.r.t.~$\preceq$, denoted by $P \vee Q$ and called the \emph{join of $P$ and $Q$}.
\end{itemize}
We will further assume that $\sets$ fulfils the following two additional properties:
\begin{align}
    \incidence[e]{P \vee Q} + \incidence[e]{P \wedge Q} \leq \incidence[e]{P} + \incidence[e]{Q}
\quad & \forall\, P, Q \in \sets,\; e \in E \tag{SM} \label{prop:submodularity}\\
    P \cap R \subseteq Q \quad & \forall\, P, Q, R \in \sets \text{ with } P \prec Q \prec R \tag{CS}\label{prop:consecutivity}
\end{align}
which are known as \emph{submodularity} and \emph{consecutivity}, respectively.

Furthermore, we assume that the requirement function $\pi$ is \emph{supermodular} w.r.t.~the lattice $(\sets, \preceq)$, i.e., 
$$\pi_{P \vee Q} + \pi_{P \wedge Q} \geq \pi_P + \pi_Q \quad \forall\, P, Q \in \sets$$ 
and \emph{monotone} w.r.t.~$\preceq$, i.e., $\pi_P \leq \pi_Q$ for all $P, Q \in \sets$ with $P \preceq Q$.

\citet{hoffman1978lattice} showed that under these assumptions (even when foregoing monotonicity of $\pi$) the polyhedron
$$\textstyle Y^+ := \left\{\rho \in \mathbb{R}_+^E \st \sum_{e \in P} \rho_e \geq \pi_P \ \forall\, P \in \sets\right\},$$
which they call \emph{lattice polyhedron}, is integral.
For the case that $\pi$ is monotone,
\citet{kornblum1978greedy} devised a two-phase (primal-dual) greedy algorithm for optimizing linear functions over $Y^+$, which was extended by \citet{frank1999increasing} to more general notions of sub- and supermodularity. 
These algorithms run in strongly polynomial time when provided with an oracle that, given $U \subseteq E$ returns the maximum element (w.r.t.~$\preceq$) of the sublattice $\mathcal{P}[U] := \{P \in \sets \st P \subseteq U\}$ along with the value of~$\pi_P$. We will call such an oracle \emph{lattice oracle}.
In this section, we prove the following decomposition result under the same assumptions as in~\citep{kornblum1978greedy,frank1999increasing}.

\begin{restatable}[$\clubsuit$]{theorem}{restateThmLattice}\label{thm:lattice}
    Let $(\sets, \preceq)$ be a submodular, consecutive lattice and let $\pi$ be monotone and supermodular with respect to $\preceq$.
    Then $(E, \sets, \pi)$ is {\starsuff}.
    Moreover, there is an algorithm that, given $\rho \in [0, 1]^E$, finds in polynomial time in $|E|$ and $\mathcal{T}$, a feasible decomposition of $\rho$ or asserts that $\rho \notin \polystar$, where $\mathcal{T}$ is the time for a call to a lattice oracle for $(\sets, \preceq)$ and $\pi$.
\end{restatable}

Our strategy for proving \cref{thm:lattice} is the following: If $\rho \in \polystar$, we can express it as a convex combination of extreme points (and possibly rays) of~$\polyplus$.
We can use the structure of these extreme points, implied by the optimality of the two-phase greedy algorithm, to construct ASCs and hence, via \cref{alg:decomposition}, a feasible decomposition of each extreme point. These can then be recomposed to a feasible decomposition for $\rho$.
The two-phase greedy algorithm allows us to carry out these steps efficiently as it implies a separation oracle for $\polystar$.

In the remainder of this section, we show how to construct an ASC for the case that $\rho \in \polystar$ is an extreme point of $\polyplus$. We start by describing the properties of extreme points implied by the correctness of the two-phase greedy algorithm.

\begin{theorem}[\citep{frank1999increasing}]
    Let $\rho \in \polyplus$. Then $\rho$ is an extreme point of $\polyplus$ if and only if there exists \mbox{$e_1, \dots, e_m \in E$} and $P_1, \dots, P_m \in \sets$ with the following properties:
\begin{enumerate}[label=(G\arabic*),start=1,align=left,leftmargin=*]
    \item $e_i \in P_i$ for all $i \in [m]$, \label{prop:greedy-e_i-in_P_i}
    \item $P_i = \max_{\succeq} \sets[E \setminus \{e_1, \dots, e_{i-1}\}]$ for all $i \in [m]$, \label{prop:greedy-P_i-max}
    \item $\pi_{P_i} > 0$ for all $i \in [m]$ and $\pi_{Q} \leq 0$ for all $Q \in \sets[E \setminus \{e_1, \dots, e_{m}\}]$, \label{prop:greedy-positive}
    \item $\rho$ is the unique solution to the linear system\\ \label{prop:greedy-solution}
    $$\begin{array}{rll}
        \sum_{e \in P_i} \rho_{e} & = \pi_{P_i} \qquad & \forall\; i \in [m],\\
        \rho_e & = 0 \qquad & \forall\; e \in E \setminus \{e_1, \dots, e_m\}.
    \end{array}$$
\end{enumerate}
\end{theorem}

We call such $e_1, \dots, e_m \in E$ and $P_1, \dots, P_m \in \sets$ fulfilling these properties a \emph{greedy support for $\rho$}. 
Indeed, note that \ref{prop:greedy-solution} implies $E_{\rho} \subseteq \{e_1, \dots, e_m\}$.
Properties~\ref{prop:greedy-e_i-in_P_i}-\ref{prop:greedy-solution} also imply that greedy supports have a special interval structure, enabling the following algorithmic and structural result.

\begin{restatable}[$\clubsuit$]{lemma}{restateLemLatticeS}\label{lem:lattice-S}
    Given a greedy support $e_1, \dots, e_m$ and $P_1, \dots, P_m$ of an extreme point $\rho$ of $\polyplus$ one can compute in time $\mathcal{O}(m)$ a set $S$ fulfilling
    \begin{align}
        S \subseteq E_{\rho} \text{ and } |S \cap P_i| = 1 \text{ for all } i \in [m]. \label{prop:good-S}
    \end{align}
\end{restatable}

The corresponding algorithm iterates through $e_1, \dots, e_m$ in reverse order and adds element $e_i$ to $S$ if it does not result in $|S \cap P_i| > 1$.
We now show that~$S$ as constructed above is indeed an ASC.

\begin{restatable}[$\clubsuit$]{theorem}{restateThmASCLattice}\label{thm:asc-lattice}
    If $S$ fulfils \eqref{prop:good-S} for the greedy suppport of an extreme point $\rho$ of $Y^+$, then $S$ is an ASC for $\pi$ and $\rho$.
\end{restatable}
\begin{proof}[sketch]
    Note that $S$ fulfils \ref{prop:cut:positive} as $S \subseteq E_\rho$ by \eqref{prop:good-S}.
    Next, we show that $S$ also fulfils \ref{prop:cut:S-tight-sets}.
    Assume by contradiction that there is $Q \in \sets^=_{\pi,\rho}$ with $|Q \cap S| > 1$.
    Without loss of generality, we can assume $Q$ to be $\succeq$-maximal with this property.
    We distinguish three cases.

    \begin{itemize}
        \item Case~1: $Q \preceq P_m$.
        Note that $e_j \notin Q$ for all $j \in [m]$ with $j < m$, as otherwise $Q \prec P_m \prec P_j$ would imply $e_j \in P_m$ by \eqref{prop:consecutivity}, contradicting \ref{prop:greedy-P_i-max}, which requires $P_m \subseteq E \setminus \{e_1, \dots, e_{m-1}\}$.
        Therefore $Q \cap S \subseteq \{e_m\}$, from which we conclude $|Q \cap S| \leq 1$.

        \item Case~2: There is $i \in [m]$ with $P_i \succeq Q \succ P_{i+1}$.
        It can be shown that \ref{prop:greedy-P_i-max} and \eqref{prop:consecutivity} imply $P_i \cap E_{\rho} \subseteq Q \cap E_{\rho}$ in this case. 
        Moreover, $P_i, Q \in \sets^=_{\pi,\rho}$ and monotonicty imply $\sum_{e \in P_i} \rho_e = \pi_{P_i} \geq \pi_{Q} = \sum_{e \in Q} \rho_e$. 
        We conclude that in fact $P_i \cap E_{\rho} = Q \cap E_{\rho}$.
        Thus $P_i \cap S = Q \cap S$ and $|Q \cap S| \leq 1$ by~\eqref{prop:good-S}.
        
        \item Case 3: There is $i \in [m]$ such that $Q \sim P_i$ (i.e., $Q$ and $P_i$ are incomparable w.r.t.~$\preceq$).
        Let~$i \in [m]$ be maximal with that property and define $Q_+ := Q \vee P_i$ and $Q_- := Q \wedge P_i$.
        Using standard uncrossing arguments we can show that
        $P_i, Q \in \sets^=_{\pi,\rho}$ implies $Q_+, Q_- \in \sets^=_{\pi,\rho}$ and $\incidence{Q_+ \cap S} + \incidence{Q_- \cap S} = \incidence{P_i \cap S} + \incidence{Q \cap S}$.
        Note that $|Q_+ \cap S| \leq 1$ by maximality of $Q \in \sets^=_{\pi,\rho}$ with $|Q \cap S| > 1$.
        We will show that $|Q_- \cap S| \leq 1$, which, using the above and $|P_i \cap S| = 1$ by~\eqref{prop:good-S}, implies $|Q \cap S| \leq |Q_+ \cap S| + |Q_- \cap S| - |P_i \cap S| \leq 1$, a contradiction.
        
        It remains to show $|Q_- \cap S| \leq 1$, for which we distinguish two subcases.
        First, if $i = m$, then $Q_- \cap S \subseteq \{e_m\}$ as shown in case~1 above.
        Second, if $i < m$, then maximality of $i$ with $P_i \sim Q$ implies $Q \succ P_{i+1}$ and therefore $P_i \succ Q_- = P_i \wedge Q \succeq P_{i+1}$.
        Thus either $Q_- = P_{i+1}$ and hence $|Q_- \cap S| = 1$ by~\eqref{prop:good-S} or $P_i \succ Q_- \succ P_{i+1}$, in which case $|Q_- \cap S| \leq 1$ by case~2 above.
    \end{itemize}
    The proof that $S$ fulfils \ref{prop:cut:S-hits-necessary-sets} follows similar lines but requires the use of some additional consequences of \ref{prop:greedy-e_i-in_P_i}-\ref{prop:greedy-solution}. \qed
\end{proof}

\section{Perfect Decompositions and Balanced Hypergraphs}
\label{sec:balanced}

In this section, we study decomposition-friendly systems, where every $\rho \in [0,1]^E$ has a perfect decomposition that attains requirements $\pi^{\rho}_P := \min \{\sum_{e \in P} \rho_e, 1\}$ for all $P \in \sets$. 
We show that such systems are characterized by absence of certain substructures that hinder perfect decomposition.

A \emph{special cycle} of $(E, \sets)$ consists of ordered subsets $C = \{e_1, \dots, e_k\} \subseteq E$ and $\mathcal{C} = \{P_1, \dots, P_k\} \subseteq \sets$  such that $P_i \cap C = \{e_i, e_{i+1}\}$ for $i \in [k]$, where we define $e_{k+1} = e_1$.
The \emph{length} of such a special cycle~$(C, \mathcal{C})$ is~$|C| = k = |\mathcal{C}|$.
A \emph{balanced hypergraph} is a system $(E, \sets)$ that does not have any special cycles of odd length.
Balanced hypergraphs were introduced by~\citet{berge1969rank} and have been studied extensively,  
see, e.g., the survey by~\citet{conforti2006balanced}. 
Our main result in this section is the~following:

\begin{restatable}[$\clubsuit$]{theorem}{restateThmBalanced}\label{thm:balanced}
    A set system $(E, \sets)$ is decomposition-friendly if and only if it is a balanced hypergraph.
    If $(E, \sets)$ is a balanced hypergraph, a perfect decomposition of $\rho \in [0,1]^E$ can be computed in polynomial time in~$|E|$.\footnote{Note that $|\sets|$ is bounded by $\mathcal{O}(|E|^2)$ for any balanced hypergraph~\citep{heller1957linear}. Thus, the stated running time holds even when $\sets$ is given explicitly.}
\end{restatable}

\begin{proof}[sketch]
    To see that every decomposition-friendly system needs to be a balanced hypergraph, consider any odd-lenth special cycle $(C, \mathcal{C})$ and observe that the marginals defined by $\rho_e = \frac{1}{2}$ for $e \in C$ and $\rho_e = 0$ for $e \in E \setminus C$ do not have a perfect decomposition. 
    The existence of perfect decompositions in balanced hypergraphs can be established by a reduction to the case $\pi^{\rho} \equiv 1$, for which a perfect decomposition can be obtained using an integrality result of \mbox{\citet{fulkerson1974balanced}}.
     \qed
\end{proof}

\bibliographystyle{splncs04nat}
\renewcommand{\bibsection}{\section*{References}}
\bibliography{decomposition}

\clearpage

\appendix

\section{Applications}
\label{sec:applications}

In this section we provide several examples for the application of the decomposition results derived in this paper: Two variants of a security game, a robust randomized weighted coverage problem, and a randomized committee election procedure with diversity constraints.

\subsection{Security Games}

The following generic security game was first described in \cite{matuschke2023decomposition-full}.
We will give two natural use cases of this security game and show how these can be solved using the results from the present paper.

\subsubsection{A Generic Security Game}
Consider the following generic security game played on a set system $(E, \sets)$.
A \emph{defender} $D$ determines a distribution $z$ for a random subset $S \subseteq E$ of elements to inspect, where the inspection of $e \in E$ incurs a cost of $c_e$.
Her goal is to deter an \emph{attacker} $A$, who observes $z$ and selects a strategy $P \in \sets$ or remains inactive (at payoff $0$).
The attacker's expected payoff for strategy $P$ is $r_P - \sum_{S : S \cap P \neq  \emptyset} z_S \cdot \beta$, where $r_P$ is the reward received for the attack using $P$ and $\beta > 0$ is a penalty incurred if $A$ is detected by $D$ (which happens if $P \cap S \neq \emptyset$).
Note that remaining inactive is $A$'s optimal response to $z$ if and only if $\sum_{S : S \cap P \neq \emptyset} z_S \geq \pi_P := \frac{r_P}{\beta}$ for all $P \in \sets$.
Thus, the polyhedron $Z_{\pi}$ describes exactly those inspection strategies for $D$ that deter the attacker from taking any action.
If $D$ wants to deter the attacker at minimum expected cost, she needs to solve the problem 
\begin{align}
    \textstyle \min_{z \in \polyDist} \sum_{S \subseteq E} \sum_{e \in S} c_e z_S. \tag{G}\label{prob:security}
\end{align}
As pointed out in~\citep{matuschke2023decomposition-full}, this problem can solved efficiently if the following conditions are met:
\begin{enumerate}[label=(\roman*),leftmargin=*]
    \item $(E, \sets, \pi)$ is {\starsuff}.
    \item We can efficiently compute feasible decompositions of marginals $\rho \in \polystar$.
    \item We can efficiently separate the constraints of $\polystar$, e.g., by solving the optimization problem $\max_{P \in \sets} \pi_P - \sum_{e \in P} \rho_e$ for given $\rho \in [0, 1]^E$.
\end{enumerate}
Indeed, by (i), problem~\eqref{prob:security} can be reformulated as $\min_{\rho \in \polystar} \sum_{e \in E} c_e \rho_e$, for which an optimal solution, i.e., optimal marginals $\rho$, can be found using (iii). Using (ii), these marginals can be turned into a feasible decomposition $z \in \polyDist$ with $\sum_{e \in E} \sum_{e \in S} c_e z_S = \sum_{e \in E} c_e \rho_e$, hence an optimal solution for problem~\eqref{prob:security}.

\bigskip

Two interesting use cases of the problem described above are the following.

\subsubsection{Suppressing Smuggling Operations}
Assume that the attacker wants to establish a smuggling network on a tree $T = (V, E)$ and has to decide which subset of links $P \in \sets := 2^E$ to actively use.
The reward for a smuggling network that operates exactly on link set $P$ is given by $r_P := \sum_{\{v, w\} \in R_P} \alpha_{\{v,w\}}$, where $R_P := \left\{\{v,w\} \subseteq V \st T[v, w] \subseteq P \right\}$, $T[v, w]$ denotes the unique $v$-$w$-path in $T$, and $\alpha_{\{vw\}} \geq 0$ denotes the reward received by $A$ for connecting nodes~$v$ and~$w$.

\paragraph{Implications of our Results.}
Note that  
\begin{align*}
    r_P + r_Q & = \textstyle \sum_{\{v, v\} \in R_{P}} \alpha_{\{v, w\}}
    + \sum_{\{v, v\} \in R_{Q}} \alpha_{\{v, w\}} \\
    & \textstyle = \sum_{\{v, v\} \in R_{P} \cup R_{Q}} \alpha_{\{v, w\}}
    + \sum_{\{v, v\} \in R_{P} \cap R_{Q}} \alpha_{\{v, w\}}\\
    & \textstyle  \leq \sum_{\{v, v\} \in R_{P \cup Q}} \alpha_{\{v, w\}}
    + \sum_{\{v, v\} \in R_{P \cap Q}} \alpha_{\{v, w\}}\\
    & = r_{P \cup Q} + r_{P \cap Q}
\end{align*}
for all $P, Q \in \sets$,
where the inequality follows from the fact that 
$\{v, w\} \in R_{P} \cup R_{Q}$ implies $T[v, w] \subseteq P \cup Q$ and that
$\{v, w\} \in R_{P} \cap R_{Q}$ implies $T[v, w] \subseteq P \cap Q$.
Therefore, the function $r$, and hence also $\pi = \frac{r}{\beta}$, is supermodular.
Moreover, given $P \in \sets$, the value $\pi_P$ can be easily determined in polynomial time in $|E|$, hence we have a value oracle for $\pi$.
By our results in \cref{sec:contrapolymatroid}, conditions (i) and (ii) above are fulfilled, and also the problem $\max_{P \in \sets} \pi_P - \sum_{e \in P} \rho_e$ from (iii) can be solved in polynomial time using submodular function minimization~\cite{schrijver2003combinatorial}.
In particular, this implies that problem~\eqref{prob:security} can be solved in polynomial time in $|E|$ in this setting.

\subsubsection{Protecting an Energy Network}
Assume that the defender wants to protect the links of an energy network, represented by a connected graph $G = (V, E)$ with a power source $v_0 \in V$, from being sabotaged. 
The attacker's reward for sabotaging link set $P \subseteq E$ is given by $r_P := \sum_{v \in V_P} \alpha_v$, where $V_P \subseteq V$ is the set of nodes disconnected from the power source $v_0$ in $E \setminus P$ and $\alpha_v \geq 0$ is the reward for disconnecting node $v \in V$.
Note that we can restrict the attacker's strategies without loss of generality to $\sets := \left\{\delta(U) \st U \subseteq V \setminus \{v_0\}\right\}$, where $\delta(U) \subseteq E$ denotes the cut induced by $U$, i.e., the set of edges with exactly one endpoint in~$U$.

\paragraph{Implications of our Results.}
We show that the situation described above is a special case of the lattice-polyhedron setting described in \cref{sec:lattice-polyhedra}.
Indeed, because $G$ is connected, there is a one-to-one correspondence between $P \in \sets$ and the corresponding $U_P \subseteq V \setminus \{v_0\}$ with $\delta(U_P) = P$.
We can thus define a partial order on $\sets$ by
$$P \preceq Q :\Leftrightarrow U_P \subseteq U_Q.$$
Note that $(\sets, \preceq)$ is a lattice, with meet $P \wedge Q = \delta(U_P \cap U_Q)$ and join $P \vee Q = \delta(U_P \cup U_Q)$.
This lattice inherits \eqref{prop:submodularity} from submodularity of the cut function~\cite{schrijver2003combinatorial}. It also fulfils \eqref{prop:consecutivity}: Consider any $P, Q, R \in \sets$ with $P \prec Q \prec R$ and let $e \in P \cap Q$.
Then $e = \{v, w\}$ for some $v \in U_P \cap U_R$ and some $w \in V \setminus (U_P \cup U_R)$.
But this implies $v \in U_Q$ and $w \in V \setminus U_Q$ because $U_P \subseteq U_Q \subseteq U_R$ and hence $e \in Q$.
Moreover, $\pi = \frac{1}{\beta} \cdot r$ is supermodular with respect to $(\sets, \preceq)$ because $r_{P \cup Q} + r_{P \cap Q} = \sum_{v \in U_P} \alpha_v + \sum_{v \in U_Q} \alpha_v = r_P + r_Q$ for all $P, Q \in \sets$ and it is monotone because $r_P = \sum_{v \in U_P} \alpha_v \leq \sum_{v \in U_Q} \alpha_v = r_Q$ for $U_P \subseteq U_Q$ by nonnegativity of $\alpha$.
Finally, note that given $F \subseteq E$, the maximum member of the sublattice $\sets[F] := \{P \in \sets \st P \subseteq F\}$ with respect to $\preceq$ is $\delta(V \setminus V^F_0)$, where $V^F_0$ is the connected component containing $v_0$ in the graph $G^F = \{V, E \setminus F\}$.
We thus can answer queries to a lattice oracle for $(\sets, \preceq)$ in time $\mathcal{O}(|E|)$.

Thus, in the setting described above, (i) and (ii) are fulfilled by \cref{thm:lattice}, and (iii) follows from \cref{lem:lattice-optimizer} in \cref{app:lattice}.
In particular, problem~\eqref{prob:security} can be solved in polynomial time in $|E|$ in this setting.

\subsection{Robust Randomized Weighted Coverage with Uncertain Objective}
In robust optimization, randomization is often used to overcome pessimistic worst-case scenarios~\citep{kawase2019submodular,kobayashi2017randomized,kawase2019randomized,matuschke2018robust}.
As an example for the applicability of decomposition results in this area, we consider the following robust randomized weighted coverage problem, which falls in the framework of \emph{randomized robust optimization with unknown objective} introduced by~\citet{kawase2019randomized}.
Let $U, E$ be finite sets, where each $e \in E$ is associated with a subset $U_e \subseteq U$. Let further $\Omega$ be a set of scenarios, where each scenario $\omega \in \Omega$ corresponds to a reward vector $r^{\omega} \in \mathbb{R}_+^U$
and a cost vector $c \in \mathbb{R}_+^E$.
A decision maker faces the problem of selecting a subset $S \subseteq E$ obtaining a profit of 
$$\textstyle f^{\omega}(R) := \sum_{u \in \bigcup_{e \in S} U_e} r^{\omega}_u  - \sum_{e \in S} c^{\omega}_e$$
for some scenario $\omega \in \Omega$ that is not known to her at the point of decision-making.

To hedge against the worst case she therefore decides to employ randomization, selecting a distribution $z \in Z := \{z \in [0,1]^{2^E} \st \sum_{S \subseteq E} z_{S} = 1\}$ so as to maximize her worst-case expected profit
$$\textstyle \min_{\omega \in \Omega} \sum_{S \subseteq E} f^{\omega}(S) z_{S}.$$

\paragraph{Implications of our Results.}
We show that our results in \cref{sec:balanced} imply that the decision maker's problem can be solved efficiently if we assume that the underlying system is a balanced hypergraph (see \cref{sec:balanced} for definition).
\begin{theorem}
    The problem 
    \begin{align}
        \textstyle \max_{z \in Z} \min_{\omega \in \Omega} \sum_{S \subseteq E} f^{\omega}(S) z_{S} \label{prob:robust} \tag{R}
    \end{align} can be solved in polynomial time in $|U|$, $|E|$, $\Omega$, and the encoding sizes of $r$ and~$c$, when restricted to instances where $(U, \mathcal{U})$ for $\mathcal{U} := \{U_e \st e \in E\}$ is a balanced hypergraph.
\end{theorem}

\begin{proof}
For each $u \in U$, define $P_u := \{e \in E \st u \in U_e\}$ and $\sets := \{P_u \st u \in U\}$. 
Consider the following linear program:
\begin{align*}
    \begin{array}{rrlr}
        \textup{(LP1)} \quad \max \ & t &&\\[8pt]
        \text{s.t.} \ & t & \leq \sum_{u \in U} r^{\omega}_u \pi_{P_u} - \sum_{e \in E} c^{\omega}_e \rho_e & \quad \forall\, \omega \in \Omega\\[8pt]
        & \pi_P & \leq \sum_{e \in P} \rho_e & \forall\, P \in \sets\\[8pt]
        & \rho & \in [0, 1]^E &\\[8pt]
        & \pi & \in [0, 1]^U & 
    \end{array}
\end{align*}

The optimal value of (LP1) is an upper bound on the optimal value of the decision maker's problem~\eqref{prob:robust}, as shown in the following lemma.

\begin{lemma}\label{lem:coverage-UB}
    Let $z \in Z$. Define $\rho_e := \sum_{S : e \in S} z_S$ for $e \in E$, $\pi_P := \sum_{S : S \cap P \neq \emptyset} z_S$, and $t := \min_{\omega \in \Omega} \sum_{S \subseteq E} f^{\omega}(S) z_S$. Then $(\rho, \pi, t)$ is a feasible solution to (LP1).
\end{lemma}
\begin{proof}
    Note that $\pi^{\rho}_{P} = \sum_{S : S \cap P \neq \emptyset} z_S \leq \sum_{S : S \cap P \neq \emptyset} |S \cap P| z_S = \sum_{e \in P} \rho_e$ for all $P \in \sets$ and that $\rho \in [0, 1]^E$ and $\pi \in [0, 1]^E$ by construction.
    Let $\omega \in \Omega$.
    Then 
    \begin{align*}
        \textstyle \sum_{u \in U} r^{\omega}_u \pi_{P_u} - \sum_{e \in E} c_e \rho_e & \textstyle  = \sum_{u \in U} r^{\omega}_u \sum_{S : u \in \bigcup_{e \in S} U_e} z_S - \sum_{e \in E} c^{\omega}_e \sum_{S : e \in S} z_S \\
        & \textstyle = \sum_{S \subseteq E} z_S \cdot \left(
        \sum_{u \in \bigcup_{e \in S} U_e} r^{\omega}_u - \sum_{e \in S} c^{\omega}_e
        \right)\\
        & \textstyle = \sum_{S \subseteq E} f^{\omega}(S) z_S \geq t
    \end{align*}
    by definition of $t$. \hfill$\diamondsuit$
\end{proof}

The following lemma shows that the converse is also true if there is an optimal solution $(\rho, \pi, t)$ such that $\rho$ has a feasible decomposition for~$(E, \sets, \pi)$.

\begin{lemma}\label{lem:coverage-LB}
    Let $(\rho, \pi, t)$ be a feasible solution to (LP1).
    Let $z$ be a feasible decomposition of $\rho$ for $(E, \sets, \pi)$.
    Then $\min_{\omega \in \Omega} \sum_{S \subseteq E} f^{\omega}(S) z_S \geq t$.
\end{lemma}

\begin{proof}
    Because $z$ is a feasible decomposition for $(E, \sets, \pi)$, we have $\rho_e = \sum_{S : e \in S} z_S$ for all $e \in E$ and $\sum_{S : u \in \bigcup_{e \in S} U_e} z_S = \sum_{S : S \cap P_u \neq \emptyset} z_S \geq \pi_{P_u}$ for all $u \in U$, which implies
    \begin{align*}
        t & \textstyle \; \geq \;  \sum_{u \in U} r^{\omega}_u \pi_{P_u} - \sum_{e \in E} c_e \rho_e \\
        & \textstyle \; \geq \; \sum_{u \in U} r^{\omega}_u \sum_{S : u \in \bigcup_{e \in S} U_e} z_S - \sum_{e \in E} c^{\omega}_e \sum_{S : e \in S} z_S \\
        & \textstyle \;=\; \sum_{S \subseteq E} z_S \cdot \left(
        \sum_{u \in \bigcup_{e \in S} U_e} r^{\omega}_u - \sum_{e \in S} c^{\omega}_e
        \right)\\
        & \textstyle \;=\; \sum_{S \subseteq E} f^{\omega}(S) z_S
    \end{align*}
    for all $\omega \in \Omega$. \hfill$\diamondsuit$
\end{proof}

In particular, \cref{lem:coverage-UB,lem:coverage-LB} imply that (LP1) is an exact formulation problem~\eqref{prob:robust} if the sytem $(E, \sets)$ is decomposition-friendly, as in this case $\rho$ has a feasible decomposition for $(E, \sets, \pi)$ for any feasible solution $(\rho, \pi, t)$ because of the constraints $\pi_P \leq \sum_{e \in P} \rho_e$ for all $P \in \sets$.
Our results in \cref{sec:balanced} imply that $(E, \sets)$ is decomposition-friendly if and only if it is a balanced hypergraph, which is the case if an only if the system $(U, \mathcal{U})$ is a balanced hypergraph.
Hence, if $(U, \mathcal{U})$ is a balanced hypergraph, we can find an optimal solution to~\eqref{prob:robust} in polynomial time in $|E|$ and $|U|$ by computing an optimal solution $(\rho, \pi, t)$ of (LP1) and then computing a perfect decomposition of $\rho$. \qed
\end{proof}

\subsection{Committee Election with Diversity Constraints}

As mentioned in \cref{sec:introduction}, randomization is also often used in the design of decision procedures in
social choice to overcome the impossibility of reconciling certain axiomatic desiderata (such as neutrality and anonymity), see, e.g., \citep{brandl2016consistent}.

To illustrate the applicability of the particular decomposition setting studied in this paper to such procedures, consider a randomized committee election procedure, in which a committee of $k$ members (e.g., a reviewing board at a university, or a citizen assembly~\citep{flanigan2021fair}) is to be determined from a set of candidates $E$.
A set of $n$ voters indicates their preferences for the committee, each voter casting a single vote for one of the candidates.
Let $n_e$ denote the number of votes cast for candidate $e \in E$.
To achieve a proportional representation of all votes (and spread the workload), each time the committee meets, it is randomly composed by $k$ members from the candidate pool, with the constraint that each candidate is included with marginal probability $\rho_e := \frac{k \cdot n_e}{n}$ (for simplicity, we assume that no candidate receives more than a $\frac{1}{k}$ fraction of the votes).
To ensure a diverse composition of the committee, different (possibly overlapping) groups of candidates $P_1, \dots, P_\ell \subseteq E$, indicating different backgrounds (such as gender, ethnicity, department membership, employment group etc.), each should be represented with a sufficiently high probability at any given meeting, say $\pi_P$ for $P \in \sets := \{P_1, \dots, P_\ell\}$.

We are thus looking for a distribution that determines a random committee~\mbox{$S \subseteq E$} of size $|S| = k$ such that $\prob{e \in S} = \rho_e$ for each $e \in E$, i.e., each candidate $e$ is in the committee with probability $\rho_e$,
and $\prob{S \cap P \neq \emptyset} \geq \pi_P$ for $P \in \sets$, i.e., each group $P$ is represented in the committee with probability at least $\pi_P$.
Note that the only difference to the decomposition problem introduced in \cref{sec:introduction} is the additional constraint that the random set $S$ must be of size~$k$. 

How to deal with such a cardinality constraint for the sets in the support of our distribution is an interesting question for future research. Here, we give a simple approximation result that shows that we can accommodate this cardinality constraint at a small loss in the probability of intersecting the sets in $\sets$ if $|\sets| \leq k$ and if each $\rho_e$ is reasonably small. In particular, we can apply the transformation described in the theorem below to any decomposition derived with the results in this paper to obtain a randomized committee election that approximately satisfies the diversity constraints if $|\sets| \leq k$, e.g., obtaining near-perfect decompositions if $(E, \sets)$ is a balanced hypergraph, as described in \cref{sec:balanced}.\footnote{Note that such a balanced hypergraph arises, e.g., when $\sets = \sets_1 \dot\cup \sets_2$, with $P \cap Q = \emptyset$ if $P, Q \in \sets_i$ for some~$i \in \{1, 2\}$, i.e., if we consider two diversity criteria where, for each criterion, each candidate belongs to at most one group.} 

\begin{theorem}\label{thm:cardinality}
Let $(E, \sets)$ be a set system, let $\pi : \sets \rightarrow (-\infty, 1]$ and let $\rho \in [0, 1]^E$ such that $\sum_{e \in E} \rho_e = k$ for some $k \in \mathbb{N}$ with $|\sets| \leq k$.
Define $\varepsilon := \max_{e \in E} \rho_e$ and $\pi'_P = (1 - \varepsilon) \pi_P$ for $P \in \sets$. 
If there is a feasible decomposition of $\rho$ for $(E, \sets, \pi)$, then there is a feasible decomposition $z^*$ of $\rho$ for $(E, \sets, \pi')$ such that $|S| = k$ for all $S \subseteq E$ with $z^*_S > 0$.
\end{theorem}

\begin{proof}
    Let $z$ be any feasible decomposition of $\rho$ for $(E, \sets, \pi)$.
    We iteratively construct a new vector $z' \in [0, 1]^{2^E}$, starting from $z' = 0$, doing the following for each $S \subseteq E$ with $z_S > 0$: 
    For each $P \in \sets$ with $P \cap S \neq \emptyset$ select an arbitrary $e_P \in P \cap S$ and let $S' := \{e_P \st P \in \sets, S \cap P \neq \emptyset\}$ and increase $z'_{S'}$ by $z_S$.
    Note that by construction, $|S| \leq |\sets| \leq k$ for each $S \subseteq E$ with $z'_S > 0$.
    
    Now define $\hat{z} \in [0, 1]^{2^E}$ by $$\hat{z}_S := \begin{cases}
    (1 - \varepsilon) \cdot z'_\emptyset + \varepsilon & \text{ for } S = \emptyset,\\
    (1 - \varepsilon) \cdot  z'_S & \text{otherwise}.\\
    \end{cases}$$
    Note that $\sum_{S \subseteq E} \hat{z}_S = \varepsilon + (1 - \varepsilon)\sum_{S \subseteq E} z'_S = 1$ and
    $$\textstyle \sum_{S : e \in S} \hat{z}_S = (1 - \varepsilon) \cdot \sum_{S : e \in S} z'_S \leq (1 - \varepsilon) \cdot \sum_{S : e \in S} z_S = (1 - \varepsilon) \cdot \rho_e$$ for all $e \in E$, and moreover,
    \begin{align*}
        \textstyle \sum_{S : S \cap P \neq \emptyset} \hat{z}_S & \textstyle = (1 - \varepsilon) \cdot \sum_{S : S \cap P \neq \emptyset} z'_S \\
        & \textstyle = (1 - \varepsilon) \cdot \sum_{S : S \cap P \neq \emptyset} z_S \geq (1 - \varepsilon) \cdot \pi_P
    \end{align*} for $P \in \sets$ by construction of $z'$.

    Let $\mathcal{S} := \{S \subseteq E \st \hat{z}_S > 0\}$.
    For $e \in E$, define $\bar{\rho}_e := \rho_e - \sum_{S : e \in S} \hat{z}_S$.
    Consider the following linear program:
    \begin{align*}
        \begin{array}{rlr}
            (\textup{LP}_k) \qquad 
            \sum_{e \in E : e \notin S} x_{e,S} & = (k - |S|) \cdot \hat{z}_S & \quad \forall\, S \in \mathcal{S}\\[8pt]
            \sum_{S \in \mathcal{S} : e \notin S} x_{e,S} & = \bar{\rho}_e & \quad \forall\, e \in E\\[8pt]
            x_{e,S} & = 0 & \quad \forall\, S \in \mathcal{S}, e \in S\\[8pt]
            x_{e,S} & \in [0, \hat{z}_S] & \forall\, e \in E, S \in \mathcal{S}
        \end{array}     
    \end{align*}

    \begin{claim}
        The linear program $(\textup{LP}_k)$ has a feasible solution.
    \end{claim}
    \begin{proof}
        Note that $(\textup{LP}_k)$ describes an instance of the feasibility version of a transportation problem and that $$\textstyle  
        \sum_{e \in E} \bar{\rho}_e = \sum_{e \in E} \rho_e - \sum_{e \in E} \sum_{S : e \in S} \hat{z}_S
        = \sum_{S \in \mathcal{S}} (k - |S|) \cdot \hat{z}_S$$
        where the second identity follows from $\sum_{e \in E} \rho_e = k$ and $\sum_{S \in \mathcal{S}} z'_S = 1$.
        Hence, feasibility of $(\textup{LP}_k)$ is equivalent to the cut condition, which in this case is equivalent to
        \begin{align*}
             \textstyle 
             \sum_{S \in W} (k - |S|) \cdot  \hat{z}_S + 
             \sum_{S \in \mathcal{S} \setminus W} |V \setminus S| \cdot \hat{z}_S
             %\sum_{e \in V} \sum_{S \in \mathcal{S} \setminus W : e \notin S} \hat{z}_S 
             \geq \sum_{e \in V} \bar{\rho}_e
        \end{align*}
        for all $V \subseteq E$ and all $W \subseteq \mathcal{S}$.
        We distinguish two cases:
        \begin{itemize}
            \item If $|V| \geq k$, then
        \begin{align*}
             \textstyle 
             \sum_{S \in W} (k - |S|) \cdot  \hat{z}_S + 
             \sum_{S \in \mathcal{S} \setminus W} |V \setminus S| \cdot \hat{z}_S 
             & \textstyle 
             \geq \sum_{S \in \mathcal{S}} (k - |S|) \cdot  \hat{z}_S
             \\ & \textstyle
             = \sum_{e \in E} \bar{\rho}_E \geq \sum_{e \in V} \bar{\rho}_e
        \end{align*}
            for all $W \subseteq \mathcal{S}$.
            \item If $|V| < k$, then 
        \begin{align*}
            \textstyle 
             \sum_{S \in W} (k - |S|) \cdot  \hat{z}_S + 
             \sum_{S \in \mathcal{S} \setminus W} |V \setminus S| \cdot \hat{z}_S 
             & \textstyle 
             \geq |V| \cdot \hat{z}_{\emptyset}
             \\ & \textstyle
             \geq |V| \cdot \varepsilon \geq \sum_{e \in V} \bar{\rho}_e 
        \end{align*}
            for all $W \subseteq \mathcal{S}$, where the final two inequalities follow from $\hat{z}_{\emptyset} \geq \varepsilon \geq \rho_{e} \geq \bar{\rho}_e$ for all $e \in E$.\hfill$\diamondsuit$
        \end{itemize}        
    \end{proof}

    Given a feasible solution $x$ to $(\textup{LP}_k)$, we obtain  $z^*$ by defining it as the distribution of a random set $S^*$ constructed as follows.
    Fix an arbitrary ordering~$e_1, \dots, e_m$ of the candidates, where $m := E$.
    For $S \in \mathcal{S}$ and $\tau \in [0, 1]$, define
    $$T_{S,\tau} := \{e_i \st \alpha_i \leq \tau + h < \alpha_{i+1} \text{ for some } h \in \mathbb{Z}_+\},$$
    where $\alpha_0 := 0$ and $\alpha_i := \alpha_{i-1} + \frac{x_{e,{S}}}{\hat{z}_{S}}$ for $i \in [m]$.
    Define the random set~$S^*$ by~$S^* := \hat{S} \cup T_{\hat{S},\tau}$ 
    where $\hat{S} \in \mathcal{S}$ is drawn from the distribution specified by $\hat{z}$ and, independently, $\tau \sim U[0, 1]$ is drawn uniformly at random from $[0, 1]$.
    We will show that the random set $S^*$ thus constructed fulfils $\prob{|S^*| = k} = 1$, $\prob{e \in S^*} = \rho_e$ for all $e \in E$ and $\prob{S^* \cap P \neq \emptyset} \geq \pi_P$ for all $P \in \sets$.
    Hence $z^*$ defined by $z^*_T := \prob{S^* = T}$ is a feasible decomposition of $\rho$ for $(E, \sets, \pi)$ with $|T| = k$ for all $T$ with $z^*_T > 0$.

    \begin{claim}
        For $\tau \sim U[0, 1]$ and any fixed $S \in \mathcal{S}$, it holds that
        \begin{enumerate}
            \item $\prob{|T_{S,\tau}| = k - |S|} = 1$,\label{cl:T-cardinality}
            \item $\prob{e_i \in T_{S,\tau}} = \alpha_{i} - \alpha_{i-1} = \frac{x_{e,S}}{\hat{z}_S}$ for all $e \in E$, and\label{cl:prob-ei}
            \item $\prob{S \cap T_{S,\tau} = \emptyset} = 1$.\label{cl:S-T-disjoint}
        \end{enumerate}
    \end{claim}
    \begin{proof}
        Note that $\alpha_{i} - \alpha_{i-1} \leq 1$ for all $i \in [m]$ and that $\alpha_m = \frac{\sum_{e \in E} x_{e, S}}{\hat{z}_{S}} = k - |S|$ for all $S \in \mathcal{S}$.
        Hence, for $0 < \tau < 1$, the set $\{\tau + h \st h \in \mathbb{Z}_+\}$ intersects with exactly $k - |S|$ distinct intervals of the form $(\alpha_{i-1}, \alpha_{i}]$ for $i \in [m]$.
        Thus $\prob{|T_{S,\tau}| = k - |S|} = 1$. 
        Moreover, $\prob{e_i \in T_{S,\tau}} = \alpha_{i} - \alpha_{i-1} = \frac{x_{e,S}}{\hat{z}_S}$ for all $e \in E$.
        The third statement follows from the second statement and $x_{e,S} = 0$ for $e \in S$. 
        \hfill$\diamondsuit$
    \end{proof}
    Thus, for the random set $\hat{S}$ drawn according to $\hat{z}$, we obtain 
    $$\prob{|S^*| = k} = \prob{|\hat{S} \cup T_{\hat{S},\tau}| = k} = 1$$ by statements~\ref{cl:T-cardinality}~and~\ref{cl:S-T-disjoint}. Moreover
    \begin{align*}
        \prob{e \in S^*} & \textstyle = \prob{e \in \hat{S} \cup T_{\hat{S},\tau}} = \prob{e \in \hat{S}} + \prob{e \notin \hat{S} \; \wedge \; e \in T_{\hat{S},\tau}}\\
        & \textstyle = \sum_{S \in \mathcal{S} : e \in S} \hat{z}_S + \sum_{S \in \mathcal{S} : e \notin S} \prob{\hat{S} = S} \cdot \prob{e \in T_{\hat{S},\tau} \big| \hat{S} = S}\\
        & \textstyle = \sum_{S \in \mathcal{S} : e \in S} \hat{z}_S + \sum_{S \in \mathcal{S} : e \notin S} \hat{z}_S \cdot \frac{x_{e,S}}{\hat{z}_S} = \rho_e
    \end{align*}
    for all $e \in E$ where the last indentity follow from feasibility of $x$ for $(\textup{LP}_k)$ and the definition of $\bar{\rho}$.
    Finally, note that $$\textstyle \prob{S^* \cap P \neq \emptyset} \geq \prob{\hat{S} \cap P \neq \emptyset} = \sum_{S : S \cap P \neq \emptyset} \hat{z}_S \geq \pi_P$$ for all $P \in \sets$. \qed
\end{proof}

\section{Decomposition Algorithm}

\subsection{Comparison with the Algorithm of \citet{dahan2021probability}}
\label{app:dahan-et-al}

\cref{alg:decomposition} can be seen as a generalization of the algorithm presented by \citet[Algorithm~1]{dahan2021probability} to compute feasible decompositions under the conservation law \eqref{eq:conservation} in directed acyclic graphs, which in turn is a generalizatoin of an earlier algortihm by \citet{hoang1994efficient} for fractional coloring problems on comparability graphs.
Indeed, both our \cref{alg:decomposition} and the one of \citet{dahan2021probability} follow the same principle of iteratively adding a new set $S$ to the support of the decomposition and increasing its probability by the maximum value that maintains feasibility of the resulting residual marginals and requirements.

An important novelty in our algorithm is the use of an explicitly defined dominance relation incorporated in the definition of admissible support candidates, providing a clear abstraction of the properties needed for the algorithm to work correctly in different settings.
Using the dominance relation, our algorithm needs to maintain only a simple underestimator $\bar{\pi}$ of the actual residual requirements for the solution constructed so far.\footnote{Given a vector $z \in [0, 1]^{2^E}$ and a requirement function $\pi$, the residual requirement for $P \in \sets$ corresponds to the value $\pi_P - \sum_{S : S \cap P \neq \emptyset} z_S$. Note that $\bar{\pi}$ in our algorithm underestimates these values, as $\bar{\pi}_P$ is decreased in every iteration of our algorithm for every $P \in \sets$, even when $P \cap S \neq \emptyset$, resulting in
$\bar{\pi}_P = \pi_P - \sum_{S : S \neq \emptyset} z_S$.}
Because in our algorithm, this estimator $\bar{\pi}$ arises from the original requirements $\pi_P$ by subtracting the same constant for all $P \in \sets$, it maintains the original structure of the system $(E, \sets, \pi)$ throughout the algorithm, facilitating the construction of ASCs for a wide variety of set systems.

The algorithm of~\citet{dahan2021probability} uses a different approach, maintaining an exact account of the residual requirements~(cf.~\citep[Eq.~9]{dahan2021probability}) and eliminating certain $P \in \sets$ that are dominated with respect to a more rigid notion of dominance (implicit in~\citep[Lemma~3]{dahan2021probability}\footnote{Using the notation of the present paper, \citep[Lemma~3]{dahan2021probability} shows that when $P$ is removed from the system, there is $Q$ with $P \cap E_{\rho} \subseteq Q \cap E_{\rho}$ and $\sum_{e \in P} \rho_e - \pi_P \geq \sum_{e \in Q} \rho_e - \pi_Q$. Note that if $P$ and $Q$ fulfil these conditions, then $P \subseteq_{\pi,\rho} Q$, or $\pi_P = \pi_Q$ and $P \cap E_{\rho} = Q \cap E_{\rho}$, but the converse does not hold in general. Thus, our dominance relation is a strict relaxation of the dominance relation employed in \citep[Lemma~3]{dahan2021probability}, allowing for the ASC $S$ to intersect with fewer sets.}) from the system.
Their analysis shows that when applied on a DAG (or, equivalently, system of maximal chains of a partially ordered set), the remaining system stays closed under the $\times_e$-operation and the residual requirements fulfil \eqref{eq:conservation} on that subsystem~\citep[Proposition~2]{dahan2021probability}.
However, this approach does not seem to generalize easily to other set systems. For example, in general digraphs/abstract networks the $\times_e$-operator is not closed with respect to the sets removed and exact residual requirements may not fulfil \eqref{eq:conservation-weak}; similarly, in the lattice polyhedron setting from \cref{sec:lattice-polyhedra}, meet and join operations are not closed under the removed sets and exact residual requirements do not, in general, fulfil supermodularity with respect to the lattice.

We further remark that for the case of DAGs, our \cref{alg:decomposition}, using the ASCs presented in \cref{sec:abstract-networks}, indeed is equivalent to the algorithm of \citet{dahan2021probability}. Indeed, the rule for constructing ASCs in general abstract networks presented in \cref{sec:abstract-networks} is a generalization of the rule used to construct the set $S^k$ in \citep[Algorithm~1]{dahan2021probability}.

\subsection{Proof of \cref{lem:dominance-transitive} ($\sqsubseteq_{\pi, \rho}$ is a Partial Order)}

Recall the definition of the dominance relation $\sqsubseteq_{\pi, \rho}$: For $P, Q \in \sets$ we write $P \sqsubseteq_{\pi, \rho} Q$ if $\pi_P \leq \pi_Q - \sum_{e \in Q \setminus P} \rho_e$ and $\pi_P < \pi_Q$, or if $P = Q$.
We prove that $\sqsubseteq_{\pi, \rho}$ is indeed a partial order.

\restateLemDominanceTransitive*

\begin{proof}
    Observe that $\sqsubseteq_{\pi,\rho}$ is reflexive by definition.
Moreover, if $P \sqsubseteq_{\pi,\rho} Q$ for $P \neq Q$ then $\pi_P < \pi_Q$ and hence $Q \not\sqsubseteq_{\pi,\rho} P$. 
Thus $\sqsubseteq_{\pi,\rho}$ is also antisymmetric.
For transitivity let $P, Q, R \in \sets$ be three pairwise different paths with $P \sqsubseteq_{\pi,\rho} Q$ and $Q \sqsubseteq_{\pi,\rho} R$.
Note that
\begin{align*}
    \textstyle \pi_{P} & \textstyle \;\leq\; \pi_Q - \sum_{e \in Q \setminus P} \rho_e \\
    & \textstyle \;\leq\; \pi_R - \left(\sum_{e \in R \setminus Q} \rho_e + \sum_{e \in Q \setminus P} \rho_e\right)
    \;\leq\; \pi_R - \sum_{e \in R \setminus P} \rho_e
\end{align*}
because $R \setminus P \subseteq (R \setminus Q) \cup (Q \setminus P)$ and $\rho_e \geq 0$ for all $e \in E$.
Moreover, note that~$\pi_P < \pi_Q < \pi_R$, we conclude that $P \sqsubseteq_{\pi,\rho} R$.
Hence $\sqsubseteq_{\pi,\rho}$ also fulfills transitivity. \qed
\end{proof}

\subsection{Missing Proofs from \cref{sec:algorithm-analysis}}
\label{app:algorithm-analysis}

\restateLemCInvariants*

\begin{proof}
We prove the three statements separately.
\begin{enumerate}[label=(\alph*),leftmargin=*,align=left]
    \item 
    We prove the statement, even for $k \in K \cup \{0\}$, by induction on $k$.
    The base case $k = 0$ holds because $\rho_e^{(1)} = \rho_e \geq 0$ for $e \in E$. For the induction step with $k > 0$, note that 
    $$\textstyle \rho_e^{(k + 1)} = \rho^{(k)} - \varepsilon^{(k)} \cdot \incidence[e]{S^{(k)}}  = \rho_e - \sum_{i = 1}^{k} \varepsilon^{(i)} \cdot \incidence[e]{S^{(i)}}$$
    where the second identity uses the induction hypothesis. Moreover, note that $\rho^{(k)}_e - \varepsilon^{(k)} \cdot \incidence[e]{S^{(k)}} \geq 0$ for all $e \in E$ as $\varepsilon^{(k)} \leq \min_{e \in S^{(k)}} \rho^{(k)}_e$ by construction.
    
    \item 
    We prove the statement via induction on $k$. 
Note that the statement holds in the base case $k = 1$ because $\rho^{(1)} = \rho \in \polystar$.
For the induction step note that
\begin{align}
    \textstyle \sum_{e \in P} \rho^{(k+1)}_e & \textstyle  \;=\; \sum_{e \in P} \left( \rho^{(k)}_e - \incidence[e]{S^{(k)}} \cdot \varepsilon^{(k)} \right) \notag\\
    & \textstyle \;=\; \sum_{e \in P} \rho^{(k)}_e \ -\ |P \cap S^{(k)}| \cdot \varepsilon^{(k)}. \label{eq:cond-induction}
\end{align}
We distinguish two cases.
\begin{itemize}
    \item If $|P \cap S^{(k)}| \leq 1$, then we can apply the induction hypothesis to conclude that the right-hand side of \eqref{eq:cond-induction} is at least $\pi^{(k)}_P - \varepsilon^{(k)} = \pi^{(k+1)}_P$.
    \item If $|P \cap S^{(k)}| > 1$, then $\varepsilon^{(k)} \leq \frac{\pi^{(k)}_P - \sum_{e \in P} \rho^{(k)}_e}{1 - |P \cap S^{(k)}|}$ by construction in the algorithm.
    From this we conclude that 
    \begin{align*}
        \textstyle |P \cap S^{(k)}| \cdot \varepsilon^{(k)} & \textstyle  = \varepsilon^{(k)} + (|P \cap S^{(k)}| - 1) \cdot \varepsilon^{(k)} \\
        & \textstyle  \leq \varepsilon^{(k)} - \bigg( \pi^{(k)} - \sum_{e \in P} \rho^{(k)}_e \bigg).
    \end{align*}
    Combining this with \eqref{eq:cond-induction} yields $\sum_{e \in P} \rho^{(k+1)}_e \geq \pi^{(k)} - \varepsilon^{(k)} = \pi^{(k+1)}_P$. 
\end{itemize}

    \item 
We first show that $S^{(k)} \neq \emptyset$.
Let $\sets' := \{P \in \sets : \pi_{P}^{(k)} > 0\}$.
By \cref{lem:dominance-transitive}, there exists $P \in \sets'$ that is non-dominated in $\sets'$ with respect to $\sqsubseteq_{\pi^{(k)},\rho^{(k)}}$.
Because $S^{(k)}$ is an ASC, property~\ref{prop:cut:S-hits-necessary-sets} implies $S^{(k)} \cap P \neq \emptyset$.

We proceed to show $\varepsilon^{(k)} > 0$. Recall that $\varepsilon^{(k)} = \varepsilon_{\pi^{(k)},\rho^{(k)}}$ is the minimum of three terms.
The first term, $\min_{e \in S^{(k)}} \rho^{(k)}_e$ is positive by~\ref{prop:cut:positive}.
Also the second term $\max_{P \in \sets} \pi^{(k)}_P$ is positive by termination criterion of the while loop.
The third term is 
\begin{align*}
    \delta_{\pi^{(k)},\rho^{(k)}}(S^{(k)}) = {\textstyle \inf_{P \in \sets : |P \cap S^{(k)}| > 1}}\ 
    \frac{ \pi^{(k)}_P - \sum_{e \in P} \rho^{(k)}_e }{ 1 - |P \cap S^{(k)}| }.
\end{align*}
Note that either $\delta_{\pi^{(k)},\rho^{(k)}}(S^{(k)}) = \infty$ (in which case $\varepsilon^{(k)}$ is equal to one of the first two terms and we are done) or there is $P' \in \sets$ with $|P' \cap S^{(k)}| > 1$ that attains the infimum in the definition of $\delta_{\pi^{(k)},\rho^{(k)}}(S^{(k)})$.
Note that $|P' \cap S^{(k)}| > 1$ implies $P' \notin \sets^{=}_{\pi^{(k)},\rho^{(k)}}$ by property~\ref{prop:cut:S-tight-sets} and hence $\pi^{(k)}_{P'} - \sum_{e \in P'} \rho^{(k)}_e < 0$ by statement~\ref{inv:pi}.
Therefore $\delta_{\pi^{(k)},\rho^{(k)}}(S^{(k)}) =  \frac{ \pi^{(k)}_{P'} - \sum_{e \in P'} \rho^{(k)}_e }{ 1 - |P' \cap S^{(k)}| } > 0$. \qed
\end{enumerate}
\end{proof}

\restateLemCIterations*

\begin{proof}
    Note that in every iteration $k \in K$ at least one of the following three statements is fulfilled:
    \begin{enumerate}
        \item There is $e \in S^{(k)}$ such that $\rho^{(k)}_e = \varepsilon^{(k)}$.
        \item There is $P \in \sets$ with $\pi^{(k)}_P = \varepsilon^{(k)}$.
        \item There there is $P \in \sets$ with $|P \cap S^{(k)}| > 1$ and 
        $\frac{ \pi^{(k)}_P - \sum_{e \in P} \rho^{(k)}_e }{ 1 - |P \cap S^{(k)}|} = \varepsilon^{(k)}$.
    \end{enumerate}
    If at least one of the first two statements holds for iteration $k \in K$, we say $k$ is an \emph{element iteration}. Otherwise, we say that $k$ is an \emph{member iteration}.
    Note that if~$k$ is an element iteration of the first type, then $\rho^{(k+1)}_e = 0$ for the corresponding element $e \in S^{(k)}$.
    Moreover, if $k$ is an element iteration of the second type, then the while loop terminates after this iteration because $\pi^{(k)}_Q - \varepsilon^{(k)} \leq 0$ for all $Q \in \sets$. Moreover, such an iteration can only be preceded by $|E| - 1$ element iterations of the first type, as there must be at least one $e$ with $\rho^{(k)}_e > 0$.
    Hence there can only be $|E|$ element iterations in total.
    To bound the number of member iterations, we prove the following claim.

    \begin{claim}
        Let $k \in K$ and $e, e' \in E_{\rho^{(k)}}$.
        If there is $P \in \mathcal{P}^{=}_{\pi^{(k)}, \rho^{(k)}}$ with $e, e' \in P$, then for every $k' \geq k$ with
        $e, e' \in E_{\rho^{(k')}}$
        there is $P \in \mathcal{P}^{=}_{\pi^{(k')}, \rho^{(k')}}$ with $e, e' \in P'$.
    \end{claim}
    \begin{proof}
        We prove the claim by induction.
        The base case $k' = k$ is trivial with $P' = P$.
        For $k' > k$ with $e, e' \in E_{\rho^{(k')}}$, the induction hypothesis implies that there is $Q \in \mathcal{P}^{=}_{\pi^{(k' - 1)}, \rho^{(k' - 1)}}$ with $e, e' \in Q$ (note that $e, e' \in E_{\rho^{(k')}}$ implies $e, e' \in E_{\rho^{(k' - 1)}}$).
        Because $\sqsubseteq_{\pi^{(k'-1)},\rho^{(k'-1)}}$ is a partial order, 
        there must be a non-dominated $P' \in \sets$ with
        $Q \sqsubseteq_{\pi^{(k'-1)},\rho^{(k'-1)}} P'$ (possibly $P' = Q$).
        Note that
        $$\textstyle \sum_{q \in Q} \rho^{(k'-1)}_q = \pi^{(k'-1)}_Q \leq \pi^{(k'-1)}_{P'} - \sum_{q \in P' \setminus Q} \rho^{(k'-1)}_q,$$
        which implies $\sum_{q \in Q \cup P'} \rho^{(k'-1)}_q \leq \pi^{(k'-1)}_{P'}$. 
        Because $\sum_{q \in P'} \rho^{(k'-1)}_q \geq \pi^{(k'-1)}_{P'}$ by \cref{lem:C-invariants}\ref{inv:pi}, we conclude that in fact 
        $$\textstyle \sum_{q \in P'} \rho^{(k'-1)}_q = \pi^{(k'-1)}_{P'} \text{ and } \sum_{q \in Q \setminus P'} \rho^{(k'-1)}_q = 0.$$
        Hence $e, e' \in Q \cap P'$ and $P' \in \sets^{=}_{\pi^{(k'-1)},\rho^{(k'-1)}}$.
        Because $P'$ is non-dominated, we have $|P' \cap S^{(k'-1)}| = 1$ by property~\ref{prop:cut:S-hits-necessary-sets} of the ASC $S^{(k'-1)}$.
        Therefore, it holds that $\sum_{q \in P'} \rho^{(k')}_q = \sum_{q \in P'} \rho^{(k' - 1)}_q - \varepsilon^{(k'-1)} = \pi^{(k'-1)}_{P'} - \varepsilon^{(k'-1)} = \pi^{(k')}_{P'}$.
        We conclude that $P' \in \sets^{=}_{\pi^{(k')},\rho^{(k')}}$ as desired.
        \hfill$\diamondsuit$
    \end{proof}

        In particular, the claim implies that if $e, e' \in P$ for some $P \in \sets^{=}_{\pi^{(k)},\rho^{(k)}}$ and some $k \in K$, then $|\{e, e'\} \cap S^{(k')}| \leq 1$ for all $k' \geq k$ because of property~\ref{prop:cut:S-tight-sets} of the ASC $S^{(k')}$.
    Note that for every member iteration $k$, by definition, there are at least two distinct elements $e, e' \in S^{(k)} \subseteq E_{\rho^{(k)}}$ with $e, e' \in P$ for a member $P \in \sets$ with $\frac{ \pi^{(k)}_P - \sum_{e \in P} \rho^{(k)}_e }{ 1 - |P \cap S^{(k)}|} = \varepsilon^{(k)}$.
    Moreover, $$\textstyle\sum_{e \in P} \rho^{(k+1)}_e = \sum_{e \in P} \rho^{(k)}_e - |P \cap S^{(k)}| \varepsilon^{(k)} = \pi^{(k)}_P - \varepsilon^{(k)} = \pi^{(k+1)}_P$$ 
    implies $P \in \sets^{=}_{\pi^{(k + 1)},\rho^{(k + 1)}}$ and hence $|\{e, e'\} \cap S^{(k')}| \leq 1$ for all $k' > k$.
    In particular, the element pairs $e, e' \in S^{(k)} \cap P$ are distinct for distinct member iterations and thus there can be at most $\binom{|E|}{2}$ member iterations. \qed
\end{proof}

\restateLemCTotalEpsilon*

\begin{proof}
    Note that $\pi^{(k)}_P = \pi_P - \sum_{i = 1}^{k-1} \varepsilon^{(i)}$ for all $P \in P$ and $k \in K$, which in particular implies $\max_{P \in \sets} \pi^{(\ell)}_P = \max_{P \in \sets} \pi_P - \sum_{i=1}^{\ell-1} \varepsilon^{(i)}$.
    Moreover, it holds that $\varepsilon^{(\ell)} \leq \max_{P \in \sets} \pi^{(\ell)}_P$ by construction of $\varepsilon^{(\ell)}$, as well as $\pi^{(\ell)}_P - \varepsilon^{(\ell)} \leq 0$
    by the termination criterion of the while loop.
    Putting this together, we obtain
    $\varepsilon^{(\ell)} = \max_{P \in \sets} \pi^{(\ell)}_P = \max_{P \in \sets} \pi_P - \sum_{i=1}^{\ell-1} \varepsilon^{(i)}$, which implies the statement of the lemma. \qed
\end{proof}

\section{Lattice Polyhedra (Missing Details from \cref{sec:lattice-polyhedra})}
\label{app:lattice}

We first provide complete proofs of \cref{lem:lattice-S} and \cref{thm:asc-lattice}.
We then show how these results imply that we can in fact compute feasible decompositions for extreme points of $\polyplus$ using \cref{alg:decomposition}.
Finally, we provide a complete proof of \cref{thm:lattice}, showing how to obtain feasible decompositions for arbitrary~\mbox{$\rho \in \polystar$}.

\subsection{Proof of \cref{lem:lattice-S} and \cref{thm:asc-lattice}}
Throughout this section, we will use the notation $P \sim Q$ to indicate that $P$ and $Q$ are incomparable with respect to $\preceq$, i.e., neither $P \preceq Q$ nor $Q \preceq P$.
We will use the following properties, which are a direct consequence of properties~\ref{prop:greedy-e_i-in_P_i}-\ref{prop:greedy-solution}.

\begin{lemma}[\cite{frank1999increasing}]
    Let $\rho$ be an extreme point of $Y^+$.
    Then any greedy support for $\rho$ fulfils the following properties:
    \begin{enumerate}[label=(G\arabic*),start=5,align=left,leftmargin=*]
        \item $P_1 \succ \ldots \succ P_m$,
        \label{prop:greedy-chain}
        \item $e_i \notin P_j$ for all $i, j \in [m]$ with $i < j$, \label{prop:greedy-e_i}
        \item $P_i \cap \{e_1, \dots, e_m\} \subseteq Q \cap \{e_1, \dots, e_m\}$ for all $Q \in \sets$ and all $i \in [m]$ with $P_i \succeq Q \succeq P_{i+1}$, \label{prop:greedy-chain-containment}
        \item $e_m \in Q$ for all $Q \in \sets$ with $Q \succeq P_m$ with $\pi_Q \geq 0$. \label{prop:greedy-e_m}
    \end{enumerate}
\end{lemma}

\subsubsection{Proof of \cref{lem:lattice-S}.}
To prove \cref{lem:lattice-S}, we first state a simple greedy algorithm. In the proof of the lemma, we show that this algorithm indeed returns a set with the desired properties when provided with a greedy support.

\medskip

\begin{algorithm}[H]
  \caption{ASCs for Lattice Polyhedra}\label{alg:good-lattice}
  \setstretch{1.1}
  Let $P_1, \dots, P_m \in \sets$ and $e_1, \dots, e_m \in E$ be the greedy support for $\rho$.\\
  Initialize $S := \emptyset$.\\
  \For{$i = m$ \textup{ down to } $1$}{
    \If{$P_i \cap S = \emptyset$}{
        Set $S := S \cup \{e_{i}\}$.\\
    }
  }
  \Return $S$
\end{algorithm}

\restateLemLatticeS*

\begin{proof}
    We prove the lemma by showing that the set $S$ computed by \cref{alg:good-lattice} indeed fulfils \eqref{prop:good-S}.
    Let $S^{(k)}$ be state of the set $S$ at the end of the iteration of the for loop with $i = k$.
    Note that $S^{(k)} \subseteq \{e_k, \dots, e_m\}$ and that $S^{(1)}$ is the set returned by the algorithm. We show that the following conditions hold for all $k \in [m]$:
    \begin{enumerate}[label=(\roman*),align=left]
        \item $|S^{(k)} \cap P_j| = 1$ for all $j \in [m]$ with $j \geq k$, \label{prop:intersection-lb}
        \item $|S^{(k)} \cap P_j| \leq 1$ for all $j \in [m]$ with $j < k$, and \label{prop:intersection-ub}
        \item $S^{(k)} \subseteq E_{\rho}$. \label{prop:positive-Sk}
    \end{enumerate}
    Note that the case $k = 1$ corresponds to \eqref{prop:good-S}, hence implying the lemma.

    We show the three statements by induction on $k$, starting with the base case $k = m$. For this base case, observe that $S^{(m)} = \{e_m\}$ as $S$ is initially empty, hence $P_m \cap S^{(m)} = \{e_m\}$ and \ref{prop:intersection-lb} and \ref{prop:intersection-ub} hold for $k = m$.
    Moreover, \ref{prop:greedy-e_i} implies $P_m \cap \{e_1, \dots, e_m\} = \{e_m\}$, and hence $0 < \pi_{P_m} = \rho_{e_m}$ by \ref{prop:greedy-positive} and~\ref{prop:greedy-solution}.

    For the induction step, consider any $k \in [m-1]$ and assume that \ref{prop:intersection-lb}-\ref{prop:positive-Sk} holds for all $k' > k$.
    We distinguish two cases.
    First, assume $S^{(k+1)} \cap P_k \neq \emptyset$. In this case $S^{(k)} = S^{(k+1)}$ because $e_k$ is not added by the algorithm in iteration $i = k$. 
    The induction hypothesis directly implies that \ref{prop:intersection-lb}-\ref{prop:positive-Sk} holds for $k$.

    Second, assume $S^{(k+1)} \cap P_k = \emptyset$. In this case $S^{(k)} = S^{(k+1)} \cup \{e_k\}$ because~$e_k$ is added by the algorithm in iteration $i = k$. 
    Note that $e_k \notin P_j$ for $j > k$ and hence \ref{prop:intersection-lb} holds for $k$ by induction hypothesis.
    Moreover, note that for $j < k$ it holds that $P_j \cap S^{(k+1)} \subseteq P_k \cap S^{(k+1)} = \emptyset$ by \eqref{prop:consecutivity} because $S^{(k+1)} \subseteq \{e_{k+1}, \dots, e_m\}$.
    Thus $P_j \cap S^{(k)} \subseteq \{e_k\}$ for $j < k$, implying \ref{prop:intersection-ub} for $k$.
    Finally, let~$e'$ be the unique element in $S^{(k+1)} \cap P_{k}$ (which exists by induction hypothesis) and note that 
    \begin{align*}
        0 & \textstyle \leq \pi_{P_{k}} - \pi_{P_{k+1}} = \sum_{e \in P_k} \rho_e - \sum_{e \in P_{k+1}} \rho_e\\
        & \textstyle = \sum_{i = k}^{m} \indicator{e_i \in P_k} \cdot  \rho_{e_i} - \sum_{i = k + 1}^{m} \indicator{e_i \in P_{k+1}} \cdot \rho_{e_i}\\
        & \textstyle \leq \rho_{e_k} - \rho_{e'}
    \end{align*}
    where the first inequality follows from monotonicity and \ref{prop:greedy-chain}, the two identities follow from \ref{prop:greedy-solution} and \ref{prop:greedy-e_i} and the final inequality follows from $e_{k} \in P_{k} \setminus P_{k+1}$ and from the fact hat $e_i \in P_k$ for some $i > k$ implies $e_i \in P_{k+1}$ by \eqref{prop:consecutivity}. Thus $\rho_{e_k} \geq \rho_{e'} > 0$ by induction hypothesis, showing that \ref{prop:positive-Sk} holds for $k$. \qed
\end{proof}

\subsubsection{Proof of \cref{thm:asc-lattice}.} For a complete proof of \cref{thm:asc-lattice}, we provide the following three lemmas (note that the proof of the first lemma is also given---with some omissions---in the proof sketch for \cref{thm:asc-lattice} in the main text).

\begin{restatable}{lemma}{restateLemGreedyTightSets}\label{lem:greedy-tight-sets}
    If $S$ fulfils~\eqref{prop:good-S}, then $|Q \cap S| \leq 1$ for all $Q \in \sets^=_{\pi,\rho}$.
\end{restatable}

\begin{proof}
    Assume by contradiction that there is $Q \in \sets^=_{\pi,\rho}$ with $|Q \cap S| > 1$.
    Without loss of generality, we can assume that $Q$ is $\succeq$-maximal with this property.
    We distinguish three cases.

    \begin{itemize}
        \item Case~1: $Q \preceq P_m$.
        Note that $e_j \notin Q$ for all $j \in [m]$ with $j < m$, as otherwise $Q \prec P_m \prec P_j$ would imply $e_j \in P_m$, a contradiction to \ref{prop:greedy-e_i}.
        Therefore $Q \cap S \subseteq \{e_m\}$, from which we conclude $|Q \cap S| \leq 1$.

        \item Case~2: There is $i \in [m]$ with $P_i \succeq Q \succ P_{i+1}$.
        In this case $P_i \cap E_{\rho} \subseteq Q \cap E_{\rho}$ by \ref{prop:greedy-chain-containment}. 
        Because $\sum_{e \in P_i} \rho_e = \pi_{P_i} \geq \pi_{Q} = \sum_{e \in Q} \rho_e$ by monotonicity and $P_i, Q \in \sets^=_{\pi,\rho}$, we conclude that in fact $P_i \cap E_{\rho} = Q \cap E_{\rho}$.
        In particular $P_i \cap S = Q \cap S$ and thus $|Q \cap S| \leq 1$ by~\eqref{alg:good-lattice}.
        
        \item Case 3: There is $i \in [m]$ with $Q \sim P_i$.
        We let $i \in [m]$ be maximal with that property.
        Let $Q_+ := Q \vee P_i$ and $Q_- := Q \wedge P_i$.
        Using $\rho \in \polyplus$, submodularity of the lattice, and $P_i, Q \in \sets^=_{\pi,\rho}$, we obtain
        \begin{align*}
            \textstyle \pi_{Q_+} + \pi_{Q_-} \leq \sum_{e \in Q_+} \rho_e + \sum_{e \in Q_-} \rho_e \leq \sum_{e \in P_i} \rho_e + \sum_{e \in Q} \rho_e = \pi_{P_i} + \pi_{Q}.
        \end{align*}
        Note that furthermore $\pi_{P_i} + \pi_{Q} \leq \pi_{Q_+} + \pi_{Q_-}$ by supermodularity of $\pi$ and hence the above must hold with equality everywhere.
        Thus, in particular $Q_+, Q_- \in \sets^=_{\pi,\rho}$ and $\incidence{Q_+ \cap S} + \incidence{Q_- \cap S} = \incidence{P_i \cap S} + \incidence{Q \cap S}$ as $\rho_e > 0$ for all $e \in S$.
        Note that $|Q_+ \cap S| \leq 1$ by maximality of $Q \in \sets^=_{\pi,\rho}$ with $|Q \cap S| > 1$.
        We show that $|Q_- \cap S| \leq 1$, which, using the above and $|P_i \cap S| = 1$ by~\eqref{alg:good-lattice}, implies $|Q \cap S| \leq |Q_+ \cap S| + |Q_- \cap S| - |P_i \cap S| \leq 1$, a contradiction.

        It remains to show $|Q_- \cap S| \leq 1$, for which we distinguish two subcases.
        First, if $i = m$, then $Q_- \cap S \subseteq \{e_m\}$ as shown in case~1 above.
        Hence $Q_- \cap S = P_i \cap S$ in this case.
        Second, if $i < m$, then maximality of $i$ with $P_i \sim Q$ implies $Q \succ P_{i+1}$ and therefore $P_i \succ Q_- = P_i \wedge Q \succeq P_{i+1}$.
        Thus either $Q_- = P_{i+1}$ and hence $|Q_- \cap S| = 1$ by~\eqref{alg:good-lattice} or $P_i \succ Q_- \succ P_{i+1}$, in which case $|Q_- \cap S| \leq 1$ as shown in case~2 above. \qed
    \end{itemize}
\end{proof}

\begin{lemma}\label{lem:greedy-greater-P_m}
    If $S$ fulfils~\eqref{prop:good-S}, then $Q \cap S \neq \emptyset$ for all $Q \in \sets$ with $Q \succeq P_m$.
\end{lemma}

\begin{proof}
    We prove the lemma inductively, showing that $Q \cap S \neq \emptyset$ for $Q \in \sets$ if $Q' \cap S$ for all $Q' \in \sets$ with $Q \succ Q' \succeq P_m$.
    The base case is $Q = P_m$, for which $Q \cap S \neq \emptyset$ follows directly from~\eqref{prop:good-S}.
    
    Now let $Q \in \sets$ with $Q \succ P_m$ and let $i := \max \{j \in [m] \st P_j \succeq Q\}$, which exists and is at most $m-1$ because $P_1 = \max_{\succeq} \sets \succeq Q \succ P_m$.
    We distinguish two cases.
    \begin{itemize}
        
        \item Case~1: $Q \succ P_{i+1}$.
        In this case $P_i \cap S \subseteq P_i \cap \{e_1, \dots, e_m\} \subseteq Q$ by \ref{prop:greedy-chain-containment}.
        We conclude that $Q \cap S \neq \emptyset$ as $P_i \cap S \neq \emptyset$ by~\eqref{prop:good-S}.
        
        \item Case~2: $Q \sim P_{i+1}$.
        Define $Q_+ := Q \vee P_{i+1}$ and $Q_- := Q \wedge P_{i+1}$.
        Note that $P_i \succeq Q_+ \succ P_{i+1}$ and hence $\emptyset \neq P_i \cap S \subseteq Q_+$ by \ref{prop:greedy-chain-containment} and~\eqref{prop:good-S}.
        Furthermore, $Q \succ Q_- = Q \wedge P_{i+1} \succeq P_m$ as $P_m \preceq Q$ and $P_m \preceq P_{i+1}$.
        Thus $Q_- \cap S \neq \emptyset$ by induction hypothesis.
        Now submodularity yields
        $|P_{i+1} \cap S| + |Q \cap S| \geq |Q_+ \cap S| + |Q_- \cap S| \geq 2$, from which we conclude $|Q \cap S| \geq 1$ because $|P_{i+1} \cap S| = 1$ by~\eqref{prop:good-S}. 
        \qed
    \end{itemize}
\end{proof}

\begin{lemma}\label{lem:greedy-witnesses}
    If $S$ fulfils~\eqref{prop:good-S}, then $S \cap Q \neq \emptyset$ for every non-dominated $Q \in \sets$ with $\pi_Q > 0$.
\end{lemma}

\begin{proof}
    Let $Q \in \sets$ with $\pi_Q > 0$ be non-dominated with respect to $\sqsubseteq_{\pi,\rho}$.
    We distinguish three cases.
    
    \begin{itemize}
        \item Case~1: $Q \succeq P_m$.
        Then $Q \cap S \neq \emptyset$ by \cref{lem:greedy-greater-P_m}.
        
        \item Case~2: $Q \prec P_m$.
        Then $e_m \in Q$ by \ref{prop:greedy-e_m} and hence $P_m \cap S = \{e_m\} = Q \cap S$ by~\eqref{prop:good-S}.
        
        \item Case~3: $Q \sim P_m$.
        Define $Q_+ := Q \vee P_m$ and $Q_- := Q \wedge P_m$.
        Note that $Q_+ \cap S \neq \emptyset$ by \cref{lem:greedy-greater-P_m} because $Q_+ \succ P_m$.
        We distinguish two subcases.
        
        First, assume $\pi_{Q_-} \leq 0$. 
        Then $\pi_{Q} \leq \pi_{Q_+} + \pi_{Q_-} - \pi_{P_i} \leq \pi_{Q_+} - \sum_{e \in P_i} \rho_e$, where the first inequality is due to supermodularity and the second inequality follows from $\sum_{e \in P_i} \rho_e = \pi_{P_i}$ and $\pi_{Q_-} \leq 0$.
        Note that this implies $\pi_{Q} < \pi_{Q_+}$ because $\sum_{e \in P_i} \rho_e = \pi_{P_i} > 0$ and that moreover $Q_+ \setminus Q \subseteq P_i$ because $Q_+ \subseteq P_i \cup Q$ by~\eqref{prop:submodularity}.
        We thus conclude $Q \sqsubseteq_{\pi,\rho} Q_+$, a contradiction.

        Second, assume $\pi_{Q_-} > 0$.
        Then $P_m \cap S = \{e_m\} = Q_- \cap S$ by \ref{prop:greedy-e_m} because~$Q_- \prec P_m$.
        Combining this with \eqref{prop:submodularity} yields $\emptyset \neq Q_+ \cap S \subseteq Q \cap S$. \qed
    \end{itemize}
\end{proof}

\restateThmASCLattice*
\begin{proof}
    Note that $S$ fulfils \ref{prop:cut:positive} as $S \subseteq E_\rho$ by \eqref{prop:good-S}.
    Moreover, \cref{lem:greedy-tight-sets} implies that $S$ fulfils \ref{prop:cut:S-tight-sets}.
    Finally, \cref{lem:greedy-witnesses} implies that $S$ fulfils \ref{prop:cut:S-hits-necessary-sets}.
\end{proof}

\subsection{Computing Feasible Decompositions for Extreme Points}

In this section, we show that \cref{thm:asc-lattice,lem:lattice-S} imply that we can in fact use \cref{alg:decomposition} to obtain feasible decompositions of $\rho \in [0, 1]^E$ when $\rho$ is an extreme point of $\polyplus$, hence establishing the following theorem.

\begin{theorem}\label{thm:lattice-extreme-points}
    Let $\rho$ be an extreme point of $\polyplus$.
    Then $\rho$ has a feasible decomposition for $(E, \sets, \pi)$.
    Moreover, there is an algorithm that, given an extreme point $\rho$ of $\polyplus$ and a greedy support for $\rho$, finds in polynomial time in $|E|$ and $\mathcal{T}$, a feasible decomposition of $\rho$, where $\mathcal{T}$ is the time for a call to a lattice oracle for $(\sets, \preceq)$ and $\pi$.
\end{theorem}

We first show that $\pi_P \leq 1$ for all $P \in \sets$ implies that the extreme points of~$\polyplus$ are indeed contained in~$\polystar = \polyplus \cap [0, 1]^E$.

\begin{lemma}\label{lem:extremepoints-1}
    Let $y$ be an extreme point of $\polyplus$. Then $y \in [0, 1]^E$.
\end{lemma}
\begin{proof}
    By contradiction assume that $y_{e'} > 1$ for some $e' \in E$.
    Let $$\textstyle \varepsilon := \min_{P \in \sets : e' \in P} \sum_{e \in P} y_e - \pi_P.$$
    Note that $\varepsilon > 0$ because $\pi_P \leq 1$ for all $P\ in \sets$. 
    Let $v'$ be the unit vector for $e'$, i.e., $v'_{e'} = 1$ and $v_e = 0$ for all $e \in E \setminus \{e'\}$.
    Then $y + \varepsilon v', y - \varepsilon v' \in \polyplus$, a contradiction to $y$ being an extreme point.
\end{proof}

To show that we can indeed apply \cref{alg:decomposition}, we further need to establish that the properties under which we can construct ASCs are maintained throughout the run of algorithm, i.e., if $\rho$ is an extreme point of $\polyplus$ initially, then $\bar{\rho}$ remains an extreme point for the corresponding polyhedron defined by $\bar{\pi}$ also in the next iteration. 
The following lemma shows that this is indeed the case.

\begin{restatable}{lemma}{restateLemExtremePointMaintained}\label{lem:extreme-point-maintained}
    If $S$ fulfils \eqref{prop:good-S}, then
    $\bar{\pi}$ defined by $\bar{\pi}_P := \pi_P - \varepsilon_{\pi,\rho}(S)$ for $P \in \sets$ is supermodular and monotone, and $\bar{\rho}$ defined by 
    $\bar{\rho}_e := \rho_e - \incidence[e]{S} \cdot \varepsilon_{\pi,\rho}(S)$ is an extreme point solution for $\bar{Y}^+ := \left\{y \in \mathbb{R}^E_+ \st \sum_{e \in P} y_e \geq \bar{\pi}_P \ \forall\, P \in \sets\right\}$.
\end{restatable}

\begin{proof}
    Supermodularity and monotonicty of $\bar{\pi}$ follows directly from supermodularity and monotonicity of $\pi$ as subtracting a constant does not change these properties.
    Moreover, note that $\bar{\rho} \in \polystar \subseteq \polyplus$ by \cref{lem:C-invariants}\ref{inv:pi}.
    
    We show that $\bar{\rho}$ is indeed an extreme point of $\bar{Y}^+$ by constructing a greedy support that fulfils \ref{prop:greedy-e_i-in_P_i}-\ref{prop:greedy-solution} for $\bar{\rho}$ with respect to $\bar{\pi}$.
    For this, consider the greedy support $P_1, \dots, P_m$ and $e_1, \dots, e_m$ for $\rho$ used by \cref{alg:good-lattice} to compute $S$.
    Let $\bar{m} := \max \{i \in [m] \st \bar{\pi}_{P_i} > 0\} \cup \{0\}$.
    We show that $P_1, \dots, P_{\bar{m}}$ and $e_1, \dots, e_{\bar{m}}$ is a greedy support for $\bar{\rho}$ (if $\bar{m} = 0$, we show that $\bar{\rho} = 0$, which is a greedy solution with empty support).

    Indeed, $P_1, \dots, P_{\bar{m}}$ and $e_1, \dots, e_{\bar{m}}$ inherits \ref{prop:greedy-e_i-in_P_i} and \ref{prop:greedy-P_i-max} directly from the original greedy support.
    Note further that
    \begin{align}
        \textstyle \sum_{e \in P_i} \bar{\rho}_{e} = \sum_{e \in P_i} \rho_{e} - \incidence[e]{S} \cdot \varepsilon_{\pi,\rho}(S) = \pi_{P_i} -\varepsilon_{\pi,\rho}(S) = \bar{\pi}_{P_i} \label{eq:barpi-tight}
    \end{align}
    for all $i \in [m]$.
    Moreover, $\bar{\rho}_e = \rho_e = 0$ for all $e \in E \setminus \{e_1, \dots, e_m\}$ because $S \subseteq \{e_1, \dots, e_m\}$.
    This implies $\bar{\pi}_{P_i} \geq \bar{\pi}_{P_m} = \bar{\rho}_{e_m} \geq 0$ for all $i \in [m]$ and furthermore
    $\bar{\rho}_{e_i} = 0$ for all $i > \bar{m}$ (hence $\bar{\rho} = 0$ if $\bar{m} = 0$).
    Therefore, $e_1, \dots, e_{\bar{m}}$ and $P_1, \dots, P_{\bar{m}}$ fulfil~\ref{prop:greedy-solution} for $\bar{\rho}$ and $\bar{\pi}$.
    
    It remains to show that $e_1, \dots, e_{\bar{m}}$ and $P_1, \dots, P_{\bar{m}}$ fulfils \ref{prop:greedy-positive} for $\bar{\pi}$. 
    Observe that $\bar{\pi}_{P_{i}} \geq \bar{\pi}_{P_{\bar{m}}} > 0$ by monotonicity and definition of $\bar{m}$.
    Moreover, if $\bar{m} = m$, then $\bar{\pi}_Q \leq \pi_{Q} \leq 0$ for all $Q \in \sets[E \setminus \{e_1, \dots, e_m\}]$ because~\ref{prop:greedy-positive} holds for $\pi$.
    If~$\bar{m} < m$, then $P_{\bar{m} + 1} = \max \sets[E \setminus \{e_1, \dots, e_{\bar{m}}\}]$ and thus $\bar{\pi}_Q \leq \bar{\pi}_{P_{\bar{m} + 1}} \leq 0$ for all $Q \in \sets[E \setminus \{e_1, \dots, e_{\bar{m}}\}]$ by monotonicity.
    \qed
\end{proof}

By the preceding lemma we can indeed apply \cref{alg:decomposition} using the ASCs provided by \cref{lem:lattice-S}, establishing the existence of a feasible decomposition for any extreme point $\rho$ of $\polyplus$.
Moreover, the proof of \cref{lem:extreme-point-maintained} implies that a greedy support for the extreme point $\bar{\rho}$ in the next iteration can be obtained from the greedy support for the current iteration by simply removing $P_{\bar{m}+1}, \dots, P_{m}$.
Hence we can compute the corresponding ASCs using \cref{alg:good-lattice} throughout the algorithm.
To complete the proof of \cref{thm:lattice-extreme-points}, we further need to show how to compute $\varepsilon_{\bar{\pi},\bar{\rho}}(S)$ in each iteration of the algorithm.
For this, we make use of the following insight, which implies that we can find maximum violated constraints of $\polystar$ using the lattice oracle and the two-phase greedy algorithm.

\begin{lemma}\label{lem:lattice-optimizer}
    There is an algorithm that given $\rho \in [0, 1]^E$, finds a maximizer of $\max_{P \in \sets} \pi_P - \sum_{e \in P} \rho_e$ in polynomial time in $|E|$ and $\mathcal{T}$, where $\mathcal{T}$ is the time for a call to a lattice oracle for $(\sets, \preceq)$ and $\pi$.
\end{lemma}
\begin{proof}
    Note that observe that the lattice oracle can be used to run the two-phase greedy algorihtm for $(\sets, \preceq)$ for any $\bar{\pi}$ of the form $\bar{\pi}_P = \pi_P - \lambda$ for some $\lambda \in \mathbb{R}$ (as the set returned by the oracle is independent from $\pi$).
    Hence we can optimize linear functions over the polyhedron $Y^+_{\lambda} := \{y \in \mathbb{R}^E_+ \st \sum_{e \in P} y_e \geq \pi_P - \lambda \ \forall\, P \in \sets\}$.
    This implies that we can also separate the constraints of $Y^+_{\lambda}$, i.e., given $y \in \mathbb{R}^E_+$ find a $P \in \sets$ with $\pi_P - \sum_{e \in P} y_e > \lambda$ or assert that $\max_{P \in \sets} \pi_P - \sum_{e \in P} y_e \leq \lambda$. 
    Using $y = \rho$ and applying binary search for the largest $\lambda$ such that $\rho \in Y^+_{\lambda}$, we can solve $\max_{P \in \sets} \pi_P - \sum_{e \in P} \rho_e$ as desired.~\qed
\end{proof}

\subsection{Proof of \cref{thm:lattice}}

We now provide a detailed proof of \cref{thm:lattice}.

\restateThmLattice*
\begin{proof}
    Let $\rho \in \polystar$.
    Because $\polystar \subseteq \polyplus$, there are extreme points $\rho^{(1)}, \dots, \rho^{(k)}$ of $\polyplus$ such that $\rho = \sum_{i=1}^{k} \lambda_i \rho^{(i)} + r$ for some $\lambda \in \mathbb{R}_+^k$ with $\sum_{i = 1}^{k} \lambda_i = 1$ and some ray $r$ of $\polyplus$ (possibly $r = 0$). Note that $r \geq 0$ as $y \geq 0$ for all $y \in \polyplus$.
    By \cref{thm:lattice-extreme-points}, the extreme points $\rho^{(1)}, \dots, \rho^{(k)}$ have feasible decompositions $z^{(1)}, \dots, z^{(k)} \in \polyDist$ for $(E, \sets, \pi)$.
    Let $z' := \sum_{i=1}^{k} \lambda_i z^{(i)}$ and note that $z' \in \polyDist$ and 
    $$\textstyle \sum_{S : e \in S} z_S = \sum_{i=1}^{k} \lambda_i \rho^{(i)}_e = \rho_e - r_e \leq \rho_e$$
    for all $e \in E$.
    Hence, we can obtain a feasible decomposition for $\rho$ by applying \cref{lem:feasible-relax} to $z'$.

    To show that these steps can also be carried out efficiently, note that we can use the two-phase greedy algorithm to optimize linear functions over $\polyplus$.
    Hence we can verify in oracle-polynomial time if $\rho \in \polystar$ by the equivalence of optimization and separation, and obtain the extreme points $\rho^{(1)}, \dots, \rho^{(k)}$ and the corresponding coefficients using Carath\'{e}odory's algorithm~\citep[Corollary~6.5.13]{grotschel2012geometric}.
    Furthermore, as we obtained the extreme points using the two-phase greedy algorithm, we also obtained the corresponding greedy supports and we can apply \cref{thm:lattice-extreme-points} to obtain the desired feasible decompositions in polynomial time. \qed
\end{proof}

\section{Perfect Decompositions and Balanced Hypergraphs (Missing Details from \cref{sec:balanced})}

We provide a complete proof of \cref{thm:balanced}.
Recall that we call a system $(E, \sets)$ \emph{decomposition-friendly} if $(E, \sets, \pi)$ is {\starsuff} for all $\pi : \sets \rightarrow (-\infty, 1]$.
Also recall that we call $z$ a \emph{perfect decomposition} of $\rho$ for $\sets$ if $z$ is a feasible decomposition for $(E, \sets, \pi^{\rho})$, where $\pi^{\rho}$ is defined by $\pi^{\rho}_P := \min \left\{\sum_{e \in P} \rho, 1\right\}$ for all $P \in \sets$. 

Note that if $(E, \sets)$ is decomposition-friendly then every marginal vector $\rho \in [0, 1]^E$ has a perfect decomposition as $\rho$ fulfils \eqref{eq:cond} for $\pi^{\rho}$ by construction.
Conversely, if a system has $(E, \sets)$ has the property that every marginal vector $\rho$ has a perfect decomposition, then $(E, \sets)$ is decomposition-friendly, because any perfect decomposition of $\rho$ for $\sets$ is also a feasible decomposition $z$ for $(E, \sets, \pi)$ for any $\pi$ such that $\rho$ fulfils \eqref{eq:cond} for $\pi$, because $\sum_{S : S \cap P \neq \emptyset} z_S \geq \pi^{\rho}_P = \min \left\{ \sum_{e \in P} \rho_e, 1 \right\} \geq \pi_P$ for all $P \in \sets$ in this case.

\restateThmBalanced*

\begin{proof}
    We first show that any decomposition-friendly system is indeed a balanced hypergraph.
    We prove this by showing that the existence of an odd-length special cycle in $(E, \sets)$ implies the existence of a marginal vector that does not have a perfect decomposition for $\sets$.
    
    Indeed consider such a special cycle $(C, \mathcal{C})$ of odd length $k$ and define a marginal vector $\rho \in [0, 1]^E$ by $\rho_e := \frac{1}{2}$ for $e \in C$ and $\rho_e = 0$ for $e \in E \setminus C$.
    By contradiction assume there is a perfect decomposition $z$ of $\rho$.
    Note that $\frac{k}{2} = \sum_{e \in C} \rho_e = \sum_{S \subseteq E} |S| z_S \leq \max \{|S| \st S \subseteq E,\  z_S > 0\}$.
    Hence, there must be $S' \subseteq E$ with $z_{S'} > 0$ and $|S'| \geq \lceil\frac{k}{2}\rceil$, which implies that there is $i \in [k]$ with $e_i, e_{i+1} \in {S'}$ (with $e_{i+1} = e_1$ if $i = k$).
    But then 
    \begin{align*}
        \textstyle \sum_{S : S \cap P_i \neq \emptyset} z_S & \textstyle \leq \sum_{S : e_i \in S} z_S + \sum_{S : e_{i+1} \in S} z_S - \sum_{S : e_i, e_{i+1} \in S} z_S\\
        & \textstyle \leq 1 - z_{S'} < 1 = \pi^{\rho}_{P_i},
    \end{align*} 
    a contradiction. 

    Next, we show that any balanced hypergraph is decomposition-friendly.
    For this consider the system $(E', \sets')$ with $$E' := E \cup \{e_P \st P \in \sets\} \text{ and } \sets' := \{P \cup \{e_P\} \st P \in \sets\}.$$
    Note that $(E', \sets')$ is a balanced hypergraph, as each element $e_P$ appears in exactly one $P' \in \sets'$ and thus $e_P$ cannot appear in any special cycle.
    Therefore, the polytope
    $\textstyle Y := \left\{y \in [0,1]^{E'} \st \sum_{e \in P'} y_e  \geq 1 \  \forall\, P' \in \sets' \right\}$
    is integral~\citep{fulkerson1974balanced}.
    Note that every extreme point of $Y$ is the incidence vector of some set $T \in \mathcal{S}$, where $$\mathcal{S} := \{S \subseteq E' \st S \cap P' \neq \emptyset \ \forall\, P' \in \sets'\}.$$
    Let $\rho \in [0, 1]^E$ be any marginal vector on $E$. 
    We define $\rho' \in [0, 1]^{E'}$ by setting $\rho'_e = \rho_e$ for $e \in E$ and $\rho'_{e_P} = 1 - \pi^{\rho}_P$ for $P \in \sets$.
    Then $$\textstyle \sum_{e' \in P'} \rho'_{e'} = 1 - \pi^{\rho}_P + \sum_{e \in P} \rho_e \geq 1$$ for all $P \in \sets'$ and hence $\rho' \in Y$.
    By integrality of $Y$, we can express $\rho'$ as a convex combination of extreme points of $Y$, i.e., $\rho'_e = \sum_{T \in \mathcal{S} : e \in T} \lambda_T$ for $\lambda \in [0, 1]^{\mathcal{S}}$ with $\sum_{t \in \mathcal{S}} \lambda_T= 1$.
    Define $z \in [0, 1]^{2^E}$ by $z_S := \sum_{T \in \mathcal{S} : T \cap E = S} \lambda_T$. Note that $\sum_{S \subseteq E : e \in S} z_S = \sum_{T \in \mathcal{S} : e \in T} \lambda_T = \rho'_e = \rho_e$ for all $e \in E$ and that 
    \begin{align*}
        \textstyle \sum_{S \subseteq E : S \cap P \neq \emptyset} z_S 
        & \textstyle = \sum_{T \in \mathcal{S} : P \cap T\neq \emptyset} \lambda_T \\
        & \textstyle \geq 1 - \sum_{T \in \mathcal{S} : e_P \in T} \lambda_T \\
        & \textstyle = 1 - \rho'_{e_P} = \pi^{\rho}_P
    \end{align*} for all $P \in \sets$, where the inequality follows from $\sum_{T \in \mathcal{S}} \lambda_T = 1$ and the fact that for $T \in \mathcal{S}$ it holds sthat $P \cap T \neq \emptyset$ or $e_P \in T$.
    Hence, $x$ is a perfect decomposition of $\rho$ for $\sets$.

    Note further that we can construct the extreme point represenation of $\rho$ by applying Carath\'{e}odory's algorithm, thus obtaining $z$ in time polynomial in~$|E|$~(recall that $|\sets|$ is bounded by $\mathcal{O}(|E|^2)$).~\qed
\end{proof}

\end{document}